\documentclass[a4paper,USenglish]{lipics}

\usepackage{stmaryrd}
\usepackage{mathpartir}

\title{
Symbolic Solving of\newline
Extended Regular Expression Inequalities\newline
\textnormal{{Technical Report}}
}
\titlerunning{Regular Expression Inequalities}

\author{Matthias Keil}
\author{Peter Thiemann}
\affil{Institute for Computer Science\\University of Freiburg\\Freiburg, Germany\\
  \texttt{\{keilr,thiemann\}@informatik.uni-freiburg.de}
}

\authorrunning{M. Keil and P. Thiemann} 

\Copyright{Matthias Keil and Peter Thiemann}



\begin{document}

\maketitle

\newcommand{\todoInline}[1]{\todo[inline,color=orange!40]{#1}}

\newcommand{\todoError}[1]{\todo[inline,color=red!40]{#1}}
\newcommand{\todoNotice}[1]{\todo[inline,color=green!40]{#1}}

\newcommand{\todoApprove}[1]{\todo[inline,color=violet!40]{approve: #1}}

\newcommand{\todoInsert}[2][]{\todo[inline,color=green!40, #1]{#2}}
\newcommand{\todoExplain}[1]{\todo[inline,color=blue!40]{#1}}

\newcommand{\todoAddref}{\todo[inline,color=red!40]{Add reference.}}
\newcommand{\todoRewrite}[1]{\todo[inline,color=green!40]{#1}}

\newcommand{\todoDefine}[2][] {\todo[inline,color=red, #1]{#2}}

\newcommand\todoTask[2][]{\todo[color=purple, #1]{#2}}

%
%
%


\newcommand{\ERE}{ERE}
\newcommand{\RE}{RE}

\newcommand{\Rule}[1]{\textsc{#1}}

\newcommand{\entails}{\vdash}

\newcommand{\powerset}[1]{\wp(#1)}

\newcommand{\dom}{\textit{dom}}

\newcommand{\Label}[1]{\textit{#1}}

\newcommand\nat{\mathbb{N}}

\newcommand{\df}{:=}

\newcommand{\subseteqIH}{\overset{\textsf{IH}}{\subseteq}}
\newcommand{\supseteqIH}{\overset{\textsf{IH}}{\supseteq}}
\newcommand{\eqIH}{\overset{\textsf{IH}}{=}}


\newcommand{\symbola}{a}
\newcommand{\symbolb}{b}
\newcommand{\symbolc}{c}

\newcommand{\literall}{\seta}
\newcommand\literalset{\mathfrak{L}}

\newcommand{\seta}{A}
\newcommand{\setb}{B}
\newcommand{\setc}{C}

\newcommand{\inva}{\inv{A}}
\newcommand{\invb}{\inv{B}}
\newcommand{\invc}{\inv{C}}

\newcommand{\setu}{U}
\newcommand{\setv}{V}
\newcommand{\setw}{W}

\newcommand{\setl}{L}

\newcommand{\wordu}{u}
\newcommand{\wordv}{v}
\newcommand{\wordw}{w}

\newcommand{\regexr}{r}
\newcommand{\regexs}{s}
\newcommand{\regext}{t}

\newcommand{\regexAlphabet}{\Sigma}
\newcommand{\regexWords}{\regexAlphabet^*}
\newcommand{\regexSets}{\mathcal{U}}
\newcommand{\regexLang}{\mathcal{L}}

\newcommand{\regexWildcard}{?}
\newcommand{\regexDigit}{[0-9]}
\newcommand{\regexLChar}{[a-z]}
\newcommand{\regexUChar}{[A-Z]}
\newcommand{\regexChar}{[a-zA-Z]}
\newcommand{\regexAlpha}{[a-zA-Z0-9]}

\newcommand{\regexNull}{\emptyset}
\newcommand{\regexEmpty}{\epsilon}

\newcommand{\regexSymbol}{\symbola}
\newcommand{\regexLiteral}{\literall}
\newcommand{\regexSet}{\seta}
\newcommand{\regexInv}{\inva}

\newcommand\regexStarOp{*}
\newcommand\regexOrOp{+}
\newcommand\regexConcatOp{\cdot}
\newcommand\regexAndOp{\&}
\newcommand\regexNegOp{!}

\newcommand{\regexStar}[1]{#1^\regexStarOp}
\newcommand{\regexOr}[2]{{#1{\regexOrOp}#2}}
\newcommand{\regexConcat}[2]{{#1{\regexConcatOp}#2}}

\newcommand{\regexAnd}[2]{{#1{\regexAndOp}#2}}
\newcommand{\regexNeg}[1]{{\regexNegOp}#1}

\newcommand{\nextSymbols}{\textsf{first}}
\newcommand{\firstLiterals}{\textsf{next}}
\newcommand{\getNext}[1]{\nextSymbols(#1)}
\newcommand{\getFirst}[1]{\firstLiterals(#1)}
\newcommand{\getFirstPlus}[1]{\firstLiterals^{*}(#1)}

\newcommand{\inv}[1]{\overline{#1}}


\newcommand{\literalsubsetof}{\subseteq}
\newcommand{\literalcap}{\sqcap}
\newcommand{\literalbigcap}{\bigsqcap}
\newcommand{\literalin}{\in}
\newcommand{\literalcup}{\sqcup}
\newcommand{\literalbigcup}{\bigsqcup}

\newcommand{\literaljoin}{\Join}
\newcommand\literalleftjoin\ltimes


\newcommand{\automaton}{\mathcal{A}}

\newcommand{\automatonOf}[1]{\automaton\llbracket#1\rrbracket}
\newcommand{\dautomatonOf}[2]{\automaton^{\partial}_{#1}\llbracket#2\rrbracket}
\newcommand{\pautomatonOf}[2]{\automaton^{\Delta}_{#1}\llbracket#2\rrbracket}
\newcommand{\nautomatonOf}[2]{\automaton^{\nabla}_{#1}\llbracket#2\rrbracket}

\newcommand{\States}{\mathcal{Q}}
\newcommand{\Initial}{\mathcal{I}}
\newcommand{\Final}{\mathcal{F}}

\newcommand{\Transitions}{\delta}
\newcommand{\PTransitions}{\delta^{\Delta}}
\newcommand{\NTransitions}{\delta^{\nabla}}

\newcommand{\state}{q}

\newcommand{\states}[2]{\Transitions(#1,#2)}
\newcommand{\pstates}[2]{\PTransitions(#1,#2)}
\newcommand{\nstates}[2]{\NTransitions(#1,#2)}

%

\newcommand{\concat}{\cdot}

\newcommand{\lang}[1]{\llbracket#1\rrbracket}
\newcommand{\leftquotient}[2]{#1^{-1}#2}

\newcommand{\langl}{\mathcal{L}}

\newcommand{\nullable}{\nu}
\newcommand{\isNullable}[1]{\nullable(#1)}

\newcommand{\derivSymbol}{\partial}
\newcommand{\deriv}[2]{\partial_{#1}(#2)}

\newcommand{\nderiv}[2]{\nabla_{#1}(#2)}
\newcommand{\pderiv}[2]{\Delta_{#1}(#2)}

\newcommand{\emptyReducible}{\textsf{empty}}
\newcommand{\isEmpty}[1]{\emptyReducible(#1)}

\newcommand{\indifferentReducible}{\textsf{ind}}
\newcommand{\isIndifferent}[1]{\indifferentReducible(#1)}

\newcommand{\universalReducible}{\textsf{univ}}
\newcommand{\isUniversal}[1]{\universalReducsdible(#1)}

\newcommand{\ccCtx}{\Gamma}
\newcommand{\inCcCtx}[1]{#1\in\ccCtx}
\newcommand{\notInCcCtx}[1]{#1\not\in\ccCtx}

\newcommand{\isSuperSetOf}[2]{#1 \sqsupseteq #2}
\newcommand{\isSubSetOf}[2]{#1 \sqsubseteq #2}

\newcommand{\checkSuperSetOf}[2]{#1 \mathbin{\dot{\sqsupseteq}} #2}
\newcommand{\checkSubSetOf}[2]{#1 \mathbin{\dot{\sqsubseteq}} #2}
\newcommand\CC{\mathcal{CC}}
\newcommand\impliesCC{\vdash_\CC}

\newcommand{\isNotSuperSetOf}[2]{#1\not\sqsupseteq #2}
\newcommand{\isNotSubSetOf}[2]{#1\not\sqsubseteq #2}

\newcommand{\supersetof}{\sqsupseteq}
\newcommand{\subsetof}{\sqsubseteq}

\newcommand{\true}{\textit{true}}
\newcommand{\false}{\textit{false}}

\newcommand{\similar}{\approx}
\newcommand{\isSimilar}[2]{#1 \similar #2}

\newcommand{\norm}[2]{#1 \leadsto #2}


\newcommand{\Terms}{\Pi}
\newcommand{\term}{\pi}


\newcommand{\size}{\mathcal{S}} 
\newcommand{\sizeOf}[1]{\size(#1)}

\newcommand{\widthOf}[1]{\|#1\|}

\newcommand{\cardinalityOf}[1]{|#1|}

%

\newcommand{\RuleCCDelete}{(Cycle)}
\newcommand{\RuleCCUnfoldTrue}{(Unfold-True)} 
\newcommand{\RuleCCUnfoldFalse}{(Unfold-False)}
\newcommand{\RuleCCDisprove}{(Disprove)}

\newcommand{\RuleCCProveIdentity}{(Prove-Identity)}
\newcommand{\RuleCCProveNullable}{(Prove-Nullable)}
\newcommand{\RuleCCProveEmpty}{(Prove-Empty)}

\newcommand{\RuleCCDisproveEmpty}{(Disprove-Empty)}

\begin{abstract}
  This paper presents a new solution to the containment problem for \emph{extended regular expressions} that extends \emph{basic regular expressions} with intersection and complement operators and consider regular expressions on \emph{infinite alphabets} based on potentially \emph{infinite character sets}. Standard approaches deciding the containment do not take extended operators or character sets into account. The algorithm avoids the translation to an expression-equivalent automaton and provides a purely symbolic term rewriting systems for solving regular expressions inequalities.

  We give a new symbolic decision procedure for the containment problem based on Brzozowski's regular expression derivatives and Antimirov's rewriting approach to check containment. We generalize Brzozowski's syntactic derivative operator to two derivative operators that work with respect to (potentially infinite) representable character sets. 
\end{abstract}

\section{Introduction}
\label{sec:Introduction}

Regular expressions have many applications in the context of software
development and information technology: text processing, program analysis, compiler
construction, query processing, and so on. Modern programming languages
either come with standard libraries for regular expression processing or they
provide built-in facilities (e.g., Perl, Ruby, and JavaScript). Many
of these implementations augment the basic regular operations $\regexOrOp$,
$\regexConcatOp$, and ${}^\regexStarOp$ (union, concatenation, and
Kleene star) with enhancements like character classes and wildcard
literals, cardinalities, sub-matching, intersection, or
complement.

Regular expressions (\RE) are advantageous in these domains because they provide a concise means to
encode many interesting problems. \RE{}s are well suited for verification applications, because
there are decision procedures for many problems involving them: the word problem
($w\in\lang{\regexr}$), emptiness ($\lang{\regexr}=\emptyset$), finiteness, containment
($\lang{\regexr}\subseteq\lang{\regexs}$), and equivalence ($\lang{\regexr}=\lang{\regexs}$). Here
we let $\regexr$ and $\regexs$ range over \RE{} and write $\lang{\cdot}$ for the function that maps
a regular expression to the regular language that it denotes. There are also effective constructions
for operations like union, intersection, complement, prefixes, suffixes, etc on regular languages.

Recent applications impose new demands on operations involving
regular expressions. The Unicode character set with its more than
1.1 million code points requires the ability to deal effectively with very large
character sets and hence character classes. Similarly, formalizing
access contracts for objects in scripting languages even requires
regular expressions over an infinite alphabet: in this application,
the alphabet itself is an infinite formal language (the language of field
names) and a ``character class'' (i.e., a set of field names) is described
by a regular expression \cite{KeilThiemann2013-Proxy,HeideggerBieniusaThiemann2012-popl}. Hence, a
``character class'' may also have infinitely many elements.

We study the containment problem for regular expressions with two enhancements.
First, we consider \emph{extended regular expressions} (\ERE) that
contain intersection and complement operators beyond the standard regular operators of union, concatenation,
and Kleene star. An \ERE{} also denotes a regular language but it can be much
more concise than a standard \RE.
Second, we consider \ERE{}s on any alphabet that is presented as an effective boolean algebra. This
extension encompasses some infinite alphabets like the set of all field names in a scripting
language. 

The first enhancement is known to be decidable, but we give a new
symbolic decision procedure based on Brzozowski's regular expression
derivatives~\cite{Brzozowski1964} and Antimirov's rewriting approach
to check containment~\cite{Antimirov1995}. The
second enhancement has been studied previously \cite{Watson1996,vanNoordGerdemann2001,Veanes2013}, but 
in the context of automata and finite state transducers. It has not been investigated on the level
of regular expressions and in particular not in the context of Brzozowski's and Antimirov's work. We
give sufficient conditions to ensure applicability of our modification of Brzozowski's and
Antimirov's approach to the containment problem while retaining decidability.


\label{sec:related_work}
\subsection{Related Work}

The practical motivation for considering this extension is drawn from
the authors' work on checking access contracts for objects in a scripting language at run 
time~\cite{KeilThiemann2013-Proxy}. In that work, an access contract specifies a set of access paths
that start from a specific anchor object. An access path is a word over the field names of the
objects traversed by the path and we specify a set of such paths by a regular expression on the
field names. We claim that such a regular expression draws from an infinite alphabet because a field
name in a scripting language is an arbitrary string (of characters). For succinctness, we specify
sets of field names using a second level of regular expressions on characters.

In our implementation, checking containment is required to reduce memory consumption.
If the same object is restricted by more than one contract, 
then we apply containment checking to remove
redundant contracts. In that previous work, contracts were limited to
basic regular expressions and the field-level expressions were limited
to disjunctions of literals. Applying the results of the present paper
enables us to lift both restrictions.

The standard approach to checking regular expression containment is
via translation to finite automata, which may involve an exponential
blowup, and then construction of a simulation (or a bisimulation for
equivalence)~\cite{HopcroftKarp1971}. A related approach based on
non-deterministic automata is given
by Bonchi and Pous~\cite{BonchiPous2013}.

The exponential blowup is due to the construction of a deterministic
automaton from the regular expression.
Thompson's construction~\cite{Thompson1968}, creates a
non-deterministic finite automaton with $\regexEmpty$-transitions
where the number of states and transitions is linear to the length of
the (standard) regular expression.
Glushkov's~\cite{Glushkov1961} and McNaughton and
Yamada's~\cite{McNaughtonYamada1960} position automaton computes a
$n+1$-state non-deterministic automaton with up to $n^2$ transitions from a $n$-symbol expression. They are the first to
use the notion of a first symbol.
Brzozowski's regular expression derivatives~\cite{Brzozowski1964} directly calculate a deterministic
automaton from an \ERE. Antimirov's partial derivative approach~\cite{Antimirov1996} computes a $n+1$-state 
non-deterministic automation, but again without intersection and complement. 
We are not aware of an extension of Glushkov's algorithm to extended regular expressions. 

Owens and other have implemented an extension of Brzozowski's approach with character classes and
wildcards~\cite{OwensReppyTuron2009}. 

Antimirov~\cite{Antimirov1995} also proposes a symbolic method for solving regular expression inequalities, based on
partial derivatives, with exponential worst-case run time. His 
\emph{containment calculus} is closely related to the simulation technique used by Hopcroft and Karp~\cite{HopcroftKarp1971} for proving
equivalence of automata. In fact, a decision procedure for containment of regular expressions leads to one for
equivalence and vice versa.
Ginzburg~\cite{Ginzburg1967} gives an equivalence procedure based on
Brzozowski derivatives.
Antimirov's original work does not consider
intersection and complement. Caron and
coworkers~\cite{CaronChamparnaudMignot2011} extend Antimirov's work to
\ERE{} using antichains, but the resulting procedure is very complex compared to ours. 

A shortcoming of all existing approaches is their restriction to
finite alphabets. Supporting both
makes a significant difference in practice: an iteration over the
alphabet $\regexAlphabet$ is feasible for small alphabets, but it is
impractical for very large alphabets (e.g., Unicode) or infinite ones
(e.g., another level of regular languages as for our
contracts). Furthermore, most regular expressions used in practice
contain character sets. We apply techniques developed for symbolic
finite automata to address these issues~\cite{Veanes2013}.

\subsection{Overview}
This paper is organized as follows. In Section~\ref{sec:preliminaries}, we recall notations and concepts used in this
paper. Section~\ref{sec:characterlevel} introduces the notion of an effective boolean algebra for representing sets of
symbols abstractly. Section~\ref{sec:antimirovs-algorithm-containment} explains Antimirov's algorithm for checking
containment, which is the starting point of our work. Next, Section~\ref{sec:derivatives} defines two notions of
derivatives on regular expressions with respect to symbol sets. It continues to introduce the key notion of \emph{next
literals}, which ensures finiteness of our extension to Antimirov's algorithm. Section~\ref{sec:inequalities} contains
the heart of our extended algorithm, a deduction system that determines containment of extended regular expressions
along with a soundness proof.


This paper concludes with an appendix
with further technical details, examples, and proofs of theorems.


\section{Regular Expressions}
\label{sec:preliminaries}

An \emph{alphabet} $\regexAlphabet$ is a denumerable, potentially infinite set of
symbols. $\regexWords$ is the set of all finite words over symbols from
$\regexAlphabet$ with $\regexEmpty$ denoting the empty word.
Let $\symbola,\symbolb,\symbolc\in\regexAlphabet$
range over symbols; $\wordu,\wordv,\wordw\in\regexWords$ over words;
and
$\seta,\setb,\setc\subseteq\regexAlphabet$ over sets of symbols. 

Let $\langl, \langl'\subseteq\regexWords$ be languages.
The \emph{left quotient} of $\langl$ by a word $\wordu$, written
$\leftquotient{\wordu}{\langl}$, is the language $\{
\wordv~|~\wordu\wordv\in\langl \}$. It is immediate from the
definition that $\leftquotient{(\symbola\wordu)}{\langl} =
\leftquotient{\wordu}{(\leftquotient{\symbola}{\langl})}$ and that 
$\wordu\in\langl$ iff
$\regexEmpty\in\leftquotient{\wordu}{\langl}$. Furthermore,
$\langl\subseteq\langl'$ iff
$\leftquotient{\wordu}{\langl}\subseteq\leftquotient{\wordu}{\langl'}$
for all words $\wordu\in\regexWords$. The left quotient of one
language by another is defined by  $\leftquotient\langl{\langl'} = \{
\wordv~|~\wordu\wordv\in\langl', \wordu\in\langl \}$.
We abbreviate the concatenation of languages $\{\wordu\wordv~|~\wordu\in\langl,\wordv\in\langl'\}$
to $\regexConcat{\langl}{\langl'}$ and we write $\langl^*$ for the iteration $\regexConcat{\langl}{\langl^*}$.
We sometimes write $\overline\langl$ for the complement $\regexWords\setminus\langl$ and  
$\overline\seta$ for $\regexAlphabet\setminus\seta$.

An \emph{extended regular expression} (\ERE) on an alphabet
$\regexAlphabet$ is a syntactic phrase derivable from non-terminals $\regexr,\regexs,\regext$. It
comprises the the empty word, literals, union, concatenation, Kleene star, as well as
negation and intersection operators.
\begin{mathpar}
  \regexr,\regexs,\regext~\df{}~\regexEmpty~|~\literall
  ~|~\regexOr{\regexr}{\regexs}
  ~|~\regexConcat{\regexr}{\regexs}
  ~|~\regexStar{\regexr}~|~\regexAnd{\regexr}{\regexs}~|~\regexNeg{\regexr}
\end{mathpar}
Compared to standard definitions, a \emph{literal} is a set $\seta$ of symbols, which stands for an
abstract, possibly empty, character class. We write $\symbola$ instead of $\{\symbola\}$ for the
frequent case of a single letter literal. We consider
regular expressions up to similarity \cite{Brzozowski1964}, that is,
up to associativity and commutativity of the union operator with the empty set as identity.

The language $\lang{\regexr}\subseteq\regexWords$ of a regular
expression $\regexr$ is defined inductively by:
\begin{mathpar}
  \begin{array}{l@{~}l@{~}l}
    \lang{\regexEmpty} &=& \{\regexEmpty\}\\
    \lang{\regexSet} &=& \{\symbola~|~\symbola\in\regexSet\}\\
    &&
  \end{array}

  \begin{array}{l@{~}l@{~}l}
    \lang{\regexOr{\regexr}{\regexs}} &=& \lang{\regexr} \cup \lang{\regexs}\\
    \lang{\regexConcat{\regexr}{\regexs}} &=& \regexConcat{\lang{\regexr}}{ \lang{\regexs}} \\
    \lang{\regexStar{\regexr}} &=&
    \lang{{\regexr}}^*  \\
  \end{array}

  \begin{array}{l@{~}l@{~}l}
    \lang{\regexAnd{\regexr}{\regexs}} &=& \lang{\regexr} \cap \lang{\regexs}\\
    \lang{\regexNeg{\regexr}} &=& \overline{\lang{\regexr}} \\
    &&
  \end{array}
\end{mathpar}
For finite alphabets, $\lang\regexr$ is a regular language. For arbitrary alphabets, we
\emph{define} a language to be regular, if it is equal to $\lang\regexr$, for some {\ERE} $\regexr$.

We write $\isSubSetOf{\regexr}{\regexs}$ ($\regexr$ is
\emph{contained} in  $\regexs$) to express that $\lang{\regexr}\subseteq\lang{\regexs}$. 

The \emph{nullable} predicate $\isNullable{\regexr}$ indicates whether
$\lang{\regexr}$ contains the empty word, that is,
$\isNullable\regexr$ iff $\regexEmpty\in\lang{\regexr}$.  It is
defined inductively by: 
\begin{mathpar}
  \begin{array}{lll}
    \isNullable{\regexEmpty} &=& \true\\
    \isNullable{\regexSet} &=& \false\\
    &&
  \end{array}

  \begin{array}{lll} 
    \isNullable{\regexOr{\regexr}{\regexs}} &=& \isNullable{\regexr} \vee \isNullable{\regexs}\\
    \isNullable{\regexConcat{\regexr}{\regexs}} &=& \isNullable{\regexr} \wedge \isNullable{\regexs}\\
    \isNullable{\regexStar{\regexr}} &=& \true \\
  \end{array}

  \begin{array}{lll}
    \isNullable{\regexAnd{\regexr}{\regexs}} &=& \isNullable{\regexr} \wedge \isNullable{\regexs}\\
    \isNullable{\regexNeg{\regexr}} &=& \neg\isNullable{\regexr}\\
    &&
  \end{array}
\end{mathpar}
The \emph{Brzozowski derivative}
$\deriv{\symbola}{\regexr}$ of an expression
$\regexr$ w.r.t.\ a symbol $\symbola$ computes a regular expression
for the left quotient $\leftquotient{\symbola}{\lang{\regexr}}$ (see~\cite{Brzozowski1964}).
It is defined inductively as follows:
\begin{mathpar}
  \begin{array}{l@{~}l@{~}l}
    \deriv{\symbola}{\regexEmpty} &=& \regexNull\\
    \deriv{\symbola}{\seta} &=&
    \begin{cases}
      \regexEmpty, & \symbola\in\seta \\
      \regexNull,  & \symbola\notin\seta
    \end{cases}
    \\
    \deriv{\symbola}{\regexOr{\regexr}{\regexs}} &=& \regexOr{\deriv{\symbola}{\regexr}}{\deriv{\symbola}{\regexs}}\\
    \\
  \end{array}

  \begin{array}{l@{~}l@{~}l}
    \deriv{\symbola}{\regexConcat{\regexr}{\regexs}} &=& \begin{cases}
      \regexOr{\regexConcat{\deriv{\symbola}{\regexr}}{\regexs}}{\deriv{\symbola}{\regexs}}, &\isNullable{\regexr}\\
      \regexConcat{\deriv{\symbola}{\regexr}}{\regexs}, &\neg\isNullable{\regexr}
    \end{cases}
    \\
    \deriv{\symbola}{\regexStar{\regexr}} &=&
    \regexConcat{\deriv{\symbola}{\regexr}}{\regexStar{\regexr}} \\
    \deriv{\symbola}{\regexAnd{\regexr}{\regexs}} &=& \regexAnd{\deriv{\symbola}{\regexr}}{\deriv{\symbola}{\regexs}}\\
    \deriv{\symbola}{\regexNeg{\regexr}} &=& \regexNeg{\deriv{\symbola}{\regexr}}\\
  \end{array}
\end{mathpar}
The case for the set literal $\seta$ generalizes Brzozowski's
definition. The definition is extended to words by
$\deriv{\symbola\wordu}{\regexr}=\deriv{\wordu}{\deriv{\symbola}{\regexr}}$
and $\deriv{\regexEmpty}{\regexr}=\regexr$. Hence,
$\wordu\in\lang\regexr$ iff $\regexEmpty\in\lang{\deriv\wordu\regexr}$.

\section{Representing Sets of Symbols}
\label{sec:characterlevel}

The definition of an \ERE{} in Section~\ref{sec:preliminaries} just
states that a literal is a set of symbols
\mbox{$\seta\subseteq\regexAlphabet$}. However, to define tractable
algorithms, we require that $\seta$ is an element of an
effective boolean algebra \cite{Veanes2013} $(U, \sqcup,
\sqcap, \inv{\cdot}, \bot, \top)$ where $U \subseteq
\powerset\regexAlphabet$ is closed under the boolean operations. Here
$\sqcup$ and $\sqcap$ denote 
union and intersection of symbol sets, $\inv{\cdot}$ the complement,
and $\bot$ and $\top$ the empty set and the full set $\regexAlphabet$,
respectively. In this algebra, we need to be able to decide
equality of sets (hence the term \emph{effective}) and to represent
singleton symbols. 
\begin{itemize}
  \item For finite (small) alphabets, we may just take $U =
    \powerset{\regexAlphabet}$. A set of symbols may be enumerated
    and ranges of symbols may be represented by character classes, as
    customarily supported in regular expression
    implementations. Alternatively, a bitvector representation may be used.
  \item 
    If the alphabet is infinite (or just too large), then the boolean
    algebra of finite and cofinite sets of symbols is the basis for a
    suitable representation. That is, the set $U = \{\seta\in\powerset{\regexAlphabet}
    \mid \seta \text{ finite} \vee \inv{\seta}
    \text{ finite} \}$ is effectively closed under the boolean operations.
  \item
    In our application to checking access contracts in scripting
    languages~\cite{KeilThiemann2013-Proxy}, the alphabet itself is a set of words (the field names of
    objects) composed from another set $\Gamma$ of symbols:
    $\regexAlphabet \subseteq \powerset{\Gamma^*}$. To obtain an
    effective boolean algebra, we choose the set $U = \{ \seta \subseteq \powerset{\Gamma^*} \mid \seta \text{ is regular}\}$,
    which is effectively closed under the boolean operations.
  \item
    Sets of symbols may also be represented by formulas drawn from a
    decidable first-order theory over a (finite or infinite) alphabet.
    For example, the character range
    \texttt{[a-z]} would be represented by the formula $x\ge\verb|'a'|
    \wedge x\le\verb|'z'|$. 
    In this case, the boolean operations get mapped to the disjunction,
    conjunction, or negation of predicates; bottom and top are false and true,
    respectively. An SMT solver can decide equality and subset constraints.
    This approach has been demonstrated to be
    effective for very large character sets in the work on symbolic finite
    automata~\cite{Veanes2013}.
\end{itemize}
The rest of this paper is generic with respect to the choice
of an effective boolean algebra. 


\section{Antimirov's algorithm for checking containment}
\label{sec:antimirovs-algorithm-containment}

Given two regular expressions $\regexr$, $\regexs$, the
\emph{containment problem} asks whether
$\isSubSetOf{\regexr}{\regexs}$.
This problem is decidable using standard techniques 
from automata theory: construct a deterministic finite
automaton for $\regexAnd\regexr{\regexNeg\regexs}$ and check it for
emptiness. The drawback of this approach is the expensive construction
of the automaton. In general, this expense cannot be avoided because
problem is PSPACE-complete 
\cite{HuntRosenkrantzSzymanski1976,JiangRavikumar1993,MeyerStockmeyer1972}.

Antimirov~\cite{Antimirov1995} proposed an algorithm for deciding
containment of standard regular expressions (without intersection and
negation) that is based on rewriting of inequalities. His algorithm
has the same asymptotic complexity as the automata construction, but it can fail early and is
therefore better behaved in practice. We phrase the algorithm in
terms of Brzozowski derivatives to avoid introducing Antimirov's
notion of partial derivatives.
\begin{theorem}[{Containment~\cite[Proposition 7(2)]{Antimirov1995}}]\label{thm:containment-antimirov}
  For regular expressions $\regexr$ and $\regexs$,
  \begin{displaymath}
    \isSubSetOf{\regexr}{\regexs}
    \Leftrightarrow
    (\forall\wordu\in\regexWords)~%
    \isSubSetOf{\deriv{\wordu}\regexr}{\deriv{\wordu}\regexs}.
  \end{displaymath}
\end{theorem}
Antimirov's algorithm applies this theorem exhaustively to an
inequality $\checkSubSetOf\regexr\regexs$ (i.e., a proposed
containment) to generate
all pairs $\checkSubSetOf{\deriv\wordu\regexr}{\deriv\wordu\regexs}$ of
iterated derivatives until it finds a
contradiction or saturation. More precisely, Antimirov defines a
\emph{containment calculus} $\CC$ which works on sets $S$ of atoms, where an
atom is either an inequality $\checkSubSetOf\regexr\regexs$ or a boolean
constant $\true$ or $\false$. It consists of the rule
\RefTirName{CC-Disprove} which infers $\false$ from a trivially
inconsistent inequality and the rule \RefTirName{CC-Unfold} that applies
Theorem~\ref{thm:containment-antimirov} to generate new inequalities.
\begin{mathpar}
  \infer[CC-Disprove]
  {\isNullable\regexr \wedge \neg\isNullable\regexs}
  {\checkSubSetOf\regexr\regexs \impliesCC \false}

  \infer[CC-Unfold]
  {\isNullable\regexr \Rightarrow \isNullable\regexs}
  {\checkSubSetOf\regexr\regexs \impliesCC
    \{
      \checkSubSetOf{\deriv{\symbola}\regexr}{\deriv{\symbola}\regexs}
      \mid
      \symbola\in\regexAlphabet
    \}
  }
\end{mathpar}
An inference in the calculus for checking whether
$\isSubSetOf{\regexr_0}{\regexs_0}$ is a sequence $S_0 \impliesCC S_1 
\impliesCC S_2 \impliesCC \dots$ where $S_0 =
\{\checkSubSetOf{\regexr_0}{\regexs_0}\}$ and $S_{i+1}$ is an extension of $S_i$
by selecting an inequality in $S_i$ and adding the consequences of
applying one of the $\CC$ rules to it. That is, if
$\checkSubSetOf\regexr\regexs \in S_i$ and
$\checkSubSetOf\regexr\regexs \impliesCC S$, then $S_{i+1} = S_i \cup S$.

Antimirov argues \cite[Theorem~8]{Antimirov1995} that this algorithm is sound and complete by proving
(using Theorem~\ref{thm:containment-antimirov})
that $\isSubSetOf\regexr\regexs$ does not hold if and only if a set of
atoms containing $\false$ is derivable from
$\checkSubSetOf\regexr\regexs$.
The algorithm terminates because there are only finitely many
different inequalities derivable from $\checkSubSetOf\regexr\regexs$
using rule \RefTirName{CC-Unfold}. 

The containment calculus $\CC$ has two drawbacks. First, the
choice of an inequality for the next inference step is
nondeterministic. Second, an adaptation to a setting with an infinite alphabet seems
doomed because rule \RefTirName{CC-Unfold} requires us to compute the derivative for infinitely many
$\symbola\in\regexAlphabet$ at each application. We address the second
drawback next.

%
%

\section{Derivatives on Literals}
\label{sec:derivatives}

In this section, we develop a variant of
Theorem~\ref{thm:containment-antimirov} that enables us to define an
\RefTirName{CC-Unfold} rule that is guaranteed to add finitely many
atoms, even if the alphabet is infinite.
First, we observe that we may restrict the symbols considered
in rule \RefTirName{CC-Unfold} to initial symbols of the left hand side
of an inequality.
\begin{definition}[First]
  Let $\getNext{\regexr} \df
  \{\symbola~|~\symbola\wordw\in\lang{\regexr}\}$
  be the set of initial symbols derivable from regular expression $\regexr$.
\end{definition}
Clearly,
$(\forall\symbola\in\regexAlphabet)~\isSubSetOf{\deriv\symbola\regexr}{\deriv\symbola\regexs}$
iff
$(\forall\symbolb\in\getNext\regexr)~\isSubSetOf{\deriv\symbolb\regexr}{\deriv\symbolb\regexs}$
because $\deriv{\symbolb}{\regexr} = \regexNull$ for all
$\symbolb\notin\getNext\regexr$. Thus, \RefTirName{CC-Unfold} does not
have to consider the entire alphabet, but unfortunately
$\getNext\regexr$ may still be an infinite set of 
symbols. For that reason, we propose to compute derivatives with
respect to \emph{literals} (i.e., non-empty sets of symbols) instead
of single symbols. However, generalizing derivatives to literals has
some subtle problems.

To illustrate these problems, let us recall the specification of the Brzozowski derivative:
\begin{displaymath}
  \lang{\deriv{\symbola}{\regexr}} = \leftquotient{\symbola}{\lang{\regexr}}
\end{displaymath}
Now we might be tempted to consider 
the following naive extension of the derivative to a set of symbols
$\seta$.
\begin{align*}
  \lang{\deriv{\seta}{\regexr}}~=~%
  \leftquotient{\seta}{\lang{\regexr}}~=~%
  \bigcup_{\symbola\in\seta} \leftquotient{\symbola}{\lang{\regexr}} =
  \bigcup_{\symbola\in\seta} \lang{\deriv{\symbola}{\regexr}}
  \tag{wrong}
\end{align*}
However, this attempt at a specification yields inconsistent results. To see
why, consider the case where $\regexr = \regexNeg\regexs$. Generalizing from 
$\deriv{\symbola}{\regexNeg{\regexs}} = \regexNeg{\deriv{\symbola}{\regexs}}$,
we might try to define
$\deriv{\seta}{\regexNeg{\regexs}} \df{} \regexNeg{\deriv{\seta}{\regexs}}$.
If this definition was sensible, then~(\ref{eq:1}) and~(\ref{eq:2}) should yield
the same results:
\begin{eqnarray}
  \label{eq:1}
  \lang{\deriv{\seta}{\regexNeg\regexs}}
  &\stackrel{\textrm{(wrong)}}{=}&  \bigcup_{\symbola\in\seta} \lang{\deriv{\symbola}{\regexNeg\regexs}}
  \stackrel{\textrm{def}~\derivSymbol_\symbola}=  \bigcup_{\symbola\in\seta} \overline{\lang{\deriv{\symbola}{\regexs}}}
  \\
  \label{eq:2}
  \lang{\regexNeg{\deriv{\seta}{\regexs}}}
  &\stackrel{\textrm{def}~\derivSymbol_\symbola}{=}& \overline{\lang{{\deriv{\seta}{\regexs}}}}
  \stackrel{\textrm{(wrong)}}{=} \overline{\bigcup_{\symbola\in\seta} \lang{{\deriv{\symbola}{\regexs}}}} 
  \stackrel{\textrm{de Morgan}}{=} {\bigcap_{\symbola\in\seta} \overline{\lang{{\deriv{\symbola}{\regexs}}}}} 
\end{eqnarray}
However, we obtain a contradiction: with $\seta = \{\symbola, \symbolb\}$ and $\regexs =
\regexOr{\regexConcat\symbola\symbola}{\regexConcat\symbolb\symbolb}$,
\eqref{eq:1} yields $\regexWords$ whereas~\eqref{eq:2} yields
$\overline{\{\symbola, \symbolb\}}$, which is clearly different.

\subsection{Positive and Negative Derivatives}

To address this problem, we introduce two types of derivative operators with respect to symbol
sets. The \emph{positive derivative} $\pderiv{\seta}{\regexr}$ computes an
expression that contains the union of all $\deriv{\symbola}{\regexr}$ with
$\symbola\in\seta$, whereas the \emph{negative derivative}
$\nderiv{\seta}{\regexr}$ computes an expression contained in the intersection of
all $\deriv{\symbola}{\regexr}$ with $\symbola\in\seta$. 

The positive and negative derivative operators are defined by mutual
induction and flip at the {complement operator}. Most cases of their
definition are identical to the Brzozowski derivative
(cf. Section~\ref{sec:preliminaries}), thus we only show the cases that are
different\footnote{See also Appendix~\ref{sec:derivatives_full}.}. For all literals $\literall$ with $\lang{\literall}\neq\emptyset$:

\vspace{-\baselineskip}
\begin{minipage}[t]{0.4\textwidth}
  \begin{align*}
    \begin{array}{lll}
      \pderiv{\setb}{\regexSet} &\df{}& \begin{cases}
        \regexEmpty, & \regexSet\sqcap\setb\ne\bot \\
        \regexNull, & \textit{otherwise}
      \end{cases}\\
      \pderiv{\setb}{\regexNeg{\regexr}} &\df{}&  \regexNeg{\nderiv{\setb}{\regexr}}\\
    \end{array}
  \end{align*}
\end{minipage}
\begin{minipage}[t]{0.4\textwidth} 
  \begin{align*}
    \begin{array}{lll}
      \nderiv{\setb}{\regexSet} &\df{}& \begin{cases}
        \regexEmpty, & \overline{\regexSet}\sqcap\setb=\bot\\
        \regexNull, & \textit{otherwise}
      \end{cases}\\
      \nderiv{\setb}{\regexNeg{\regexr}} &\df{}&  \regexNeg{\pderiv{\setb}{\regexr}}\\
    \end{array}
  \end{align*}
\end{minipage}
\vspace{\baselineskip}

\noindent
For single symbol literals of the form $\setb=\{\symbola\}$, it holds
that  $\pderiv{\symbola}{\regexr}=\nderiv{\symbola}{\regexr}=\deriv{\symbola}{\regexr}$.
Derivatives with respect to the empty set are defined as
$\pderiv{\emptyset}{\regexr}=\regexNull$ and $\nderiv{\emptyset}{\regexr}=\regexWords$.

The following lemma states the connection between the derivative by a literal and the derivative by a symbol.
\begin{lemma}[Positive and negative derivatives]\label{thm:derivatives}
  For any $\regexr$ and $\setb$, it holds that:
  \begin{mathpar}
    \lang{\pderiv{\setb}{\regexr}}\supseteq\bigcup_{\symbola\in\setb} \lang{\deriv{\symbola}{\regexr}}

    \lang{\nderiv{\setb}{\regexr}}\subseteq\bigcap_{\symbola\in\setb} \lang{\deriv{\symbola}{\regexr}}    
  \end{mathpar}
\end{lemma}

\begin{proof}[Proof of Lemma~\ref{thm:derivatives}]
  Both inclusions are proved simultaneously by induction on $\regexr$. See Appendix~\ref{sec:proof-derivative}.
\end{proof}
%
%
The following examples illustrate the properties of the derivatives.
\begin{example}[Positive derivative]\label{exp:pderivatives}
  Let $\regexr$ be
  $\regexAnd{(\symbola\concat\symbolc)}{(\symbolb\concat\symbolc)}$
  and let the literal $\literall = \{\symbola,\symbolb\}$. 
  $$\pderiv{\literall}{\regexr} =
  \regexAnd{\pderiv{\literall}{\symbola\concat\symbolc}}{\pderiv{\literall}{\symbolb\concat\symbolc}}
  = \regexAnd{\symbolc}{\symbolc}
  \supersetof \regexOr{\deriv{\symbola}{\regexr}}{\deriv{\symbolb}{\regexr}} =
  \regexOr{\emptyset}{\emptyset}$$
\end{example}

\begin{example}[Negative derivative]\label{exp:nderivatives}
  Let $\regexr$ be
  $\regexOr{(\symbola\concat\symbolc)}{(\symbolb\concat\symbolc)}$ and
  let the literal $\literall = \{\symbola,\symbolb\}$. 
  $$\nderiv{\literall}{\regexr} =
  \regexOr{\nderiv{\literall}{\symbola\concat\symbolc}}{\nderiv{\literall}{\symbolb\concat\symbolc}}
  = \regexOr{\regexNull}{\regexNull}
  \subsetof \deriv{\symbola}{\regexr} \regexAndOp \deriv{\symbolb}{\regexr} =
  \symbolc \regexAndOp \symbolc$$
\end{example}
Positive (negative) derivatives yield an upper (lower)
approximation to the information expected from a derivative.
This approximation arises because we tried to define the derivative
with respect to an \emph{arbitrary} literal $\literall$. 
To obtain the precise information, we need to restrict these literals
suitably to \emph{next literals}.

\subsection{Next Literals}
\label{sec:next}

An occurrence of a literal $\literall$ in a regular expression $\regexr$ is \emph{initial} if there is
some $\symbola\in\regexAlphabet$ such that $\deriv\symbola\regexr$ reduces this occurrence. That is,
the computation of $\deriv\symbola\regexr$ involves $\deriv\symbola\literall$. Intuitively,
$\literall$ helps determine the first symbol of an element of $\lang\regexr$.
\begin{example}[Initial Literals]\label{example:next-literals} ~\\[-\baselineskip]
  \begin{enumerate}
    \item 
      Let $\regexr_1 = \{a,b\}.a^*$. Then $ \{a,b\}$ is an initial literal.
    \item 
      Let $\regexr_2 = \{a,b\}.a^* + \{b,c\}.c^*$. Then $ \{a,b\}$ and $\{b, c\}$ are initial.
  \end{enumerate}
\end{example}

\noindent
Generalizing from the first example, we might be tempted to conjecture that if $\literall$ is initial
in $\regexr$,
then $(\forall \symbola, \symbolb\in\literall)$ $\deriv\symbola\regexr =
\deriv\symbolb\regexr$. However, the second example shows that this conjecture is wrong:
$\{\symbola, \symbolb\}$ is initial in $\regexr_2$, but $\deriv\symbola{\regexr_2} = a^*$ and
$\deriv\symbolb{\regexr_2} = a^* + c^*$.

The problem with the second example is that $\{a,b\} \cap \{b,c\} \ne \emptyset$. Hence, instead of
identifying initial literals of an {\ERE} $\regexr$, we define a set $\getFirst\regexr$ of next
literals which are mutually disjoint, whose union contains $\getNext\regexr$, and where the symbols in
each literal yield the same derivative. In the second example, it must be that $\getFirst{\regexr_2}
= \{ \{a\}, \{b\}, \{c\} \}$.

It turns out that this problem arises in a number of cases when defining $\getFirst\regexr$
inductively. Hence, we define an operation $\literaljoin$ that builds a set of mutually
disjoint literals that cover the union of two sets of mutually disjoint literals.
\begin{definition}[Join]\label{def:join}
  Let $\literalset_1$ and $\literalset_2$ be two sets of mutually
  disjoint literals.
  \begin{align*}
    \literalset_1 \literaljoin \literalset_2~\df&
    \{
      (\literall_1\literalcap\literall_2),
      (\literall_1\literalcap\inv{\literalbigcup\literalset_2}),
      (\inv{\literalbigcup\literalset_1}\literalcap\literall_2)
      \mid
    \literall_1\in \literalset_1, \literall_2\in \literalset_2\}
  \end{align*}
\end{definition}
The following lemma states the properties of the join operation.
\begin{lemma}[Properties of Join]\label{lemma:properties-of-join}
  Let $\literalset_1$ and $\literalset_2$ be non-empty sets of mutually disjoint literals.
  \begin{enumerate}
    \item $\bigcup (\literalset_1 \literaljoin \literalset_2) = \bigcup \literalset_1 \cup \bigcup \literalset_2$.
    \item $(\forall \literall\ne\literall' \in \literalset_1 \literaljoin \literalset_2)$ $\literall\sqcap\literall' = \emptyset$.
    \item $(\forall \literall \in \literalset_1 \literaljoin \literalset_2)$ $(\forall \literall_i \in \literalset_i)$
      $\literall \sqcap \literall_i \ne \emptyset \Rightarrow \literall \sqsubseteq \literall_i$.
  \end{enumerate}
\end{lemma}

\begin{proof}[Proof of Lemma~\ref{lemma:properties-of-join}]
  See Appendix~\ref{sec:properties-join}.
\end{proof}

\begin{figure}[tp]
  \begin{mathpar}
    \begin{array}[t]{lll}
      \getFirst{\regexEmpty} &=& \{ \emptyset \}\\
      \getFirst{\regexSet} &=& \{\regexSet\}\\
    \end{array}

    \begin{array}[t]{lll}
      \getFirst{\regexOr{\regexr}{\regexs}} &=& \getFirst{\regexr} \literaljoin \getFirst{\regexs}\\
      \getFirst{\regexConcat{\regexr}{\regexs}} &=& \begin{cases}
        \getFirst{\regexr}\literaljoin\getFirst{\regexs}, & \isNullable{\regexr}\\
        \getFirst{\regexr}, & \neg\isNullable{\regexr}
      \end{cases}\\
      \getFirst{\regexStar{\regexr}} &=& \getFirst{\regexr}\\
      \getFirst{\regexAnd{\regexr}{\regexs}} &=&
      \getFirst{\regexr}\literalcap\getFirst{\regexs}\\
      \getFirst{\regexNeg{\regexr}} &=&
      \getFirst{\regexr} \cup \{ \literalbigcap\{\inv{\literall}~|~\literall\in\getFirst{\regexr}\} \}
    \end{array}
  \end{mathpar}
  \caption{Computing next literals.}
  \label{fig:first-literals}
\end{figure}

\noindent
Figure~\ref{fig:first-literals} contains the definition of $\getFirst\regexr$.
For $\regexEmpty$ the set of next literals consists
of the empty set.
The next literal of a literal $\literall$ is $\literall$.
The next literals of a union $\regexOr{\regexr}{\regexs}$ are computed as the join of the next literals of
$\regexr$ and $\regexs$ as explained in Example~\ref{example:next-literals}. 
The next literals of a concatenation $\regexConcat{\regexr}{\regexs}$ are the next literals of
$\regexr$ if $\regexr$ is not nullable. Otherwise, they are the join of the next literals of both
operands. 
The next literals of a Kleene star expression $\regexStar{\regexr}$ are the next literals of
$\regexr$.
For an intersection $\regexAnd{\regexr}{\regexs}$, the set of next literals is the set of all
intersections $\literall\literalcap\literall'$ of the next literals of both operands.
In this case, the join operation $\literaljoin$ is not needed because symbols that only appear in literals from one
operand can be elided. To see this, consider
$\getFirst{\regexAnd\symbola\symbolb} = \{ \{\symbola\} \sqcap
\{\symbolb\} \} = \{ \emptyset \}$ whereas $\{ \{\symbola\} \} \literaljoin 
\{\{\symbolb\} \} = \{ \emptyset, \{\symbola\}, \{\symbolb\} \}$.

The set of next literals of $\regexNeg{\regexr}$ comprises the next
literals of $\regexr$ and a new literal, which is the intersection of
the complements of all literals in $\getFirst{\regexr}$.
We might contemplate to exclude literals that contain symbols
$\symbola$ such that $\deriv\symbola\regexr$ is equivalent to
$\regexWords$, but we refrain from doing so because this equivalence
cannot be decided with a finite set of rewrite
rules~\cite{Redko1964}. 

The function $\getFirst\regexr \setminus \{\emptyset\}$ computes the
equivalence classes of a partial equivalence relation $\sim$ on
$\regexAlphabet$ such that equivalent symbols yield the same
derivative on $\regexr$. The relation is defined by $\symbola \sim \symbolb$ if there exists
$\literall\in\getFirst\regexr$ such that $\symbola\in\literall$ and
$\symbolb\in\literall$. Furthermore, the derivative by a symbol that is not part of the 
relation yields the empty set.
\begin{lemma}[Partial Equivalence]\label{lemma:partial-equivalence}
  Let $\literalset = \getFirst\regexr$.
  \begin{enumerate}
    \item\label{item:1} $(\forall \literall\in\literalset)$ $(\forall \symbola,\symbolb\in \literall)$
      $\deriv\symbola\regexr = \deriv\symbolb\regexr$
    \item $(\forall \symbola\notin \bigcup \literalset)$
      $\deriv\symbola\regexr = \regexNull$
  \end{enumerate}
\end{lemma}
\begin{proof}[Proof of Lemma~\ref{lemma:partial-equivalence}]
  See Appendix~\ref{sec:partial-equivalence}.
\end{proof}
It remains to show that $\getFirst\regexr$ covers all symbols in
$\getNext\regexr$.
\begin{lemma}[First]\label{thm:first}
  For all $\regexr$,
  $\bigcup \getFirst\regexr \supseteq \getNext\regexr$.
\end{lemma}
\begin{proof}[Proof of Lemma~\ref{thm:first}]
  See Appendix~\ref{sec:first-next}.
\end{proof}
Moreover, there are only finitely many different next literals for each regular
expression. 
\begin{lemma}[Finiteness]\label{thm:cardinality}
  For all $\regexr$,
  $|\getFirst{\regexr}|$ is finite.
\end{lemma}
\begin{proof}[Proof of Lemma~\ref{thm:cardinality}]
  By induction on $\regexr$. The base cases construct finite sets and the inductive cases build a
  finite number of combinations of the results from the subexpressions.
\end{proof}
Now, we put next literals to work.
If we only take positive or negative derivatives with respect to next literals, then the inclusions in
Lemma~\ref{thm:derivatives} turn into equalities. The result is that both the positive and the negative
derivative, when applied to a next literal $\literall$, calculate a regular expression for the left quotient $\literall^{-1}\lang{\regexr}$. 
\begin{theorem}[Left Quotient]\label{thm:firstderivative}
  For all $\regexr$, $\literall\in\getFirst{\regexr} \setminus \{\emptyset\}$, and $\symbola\in\lang{\literall}$:
  \begin{align*}
    \lang{\pderiv{\literall}{\regexr}}=\lang{\nderiv{\literall}{\regexr}}=\lang{\deriv{\symbola}{\regexr}}
  \end{align*}
\end{theorem}
\begin{proof}[Proof of Lemma~\ref{thm:firstderivative}]
  By induction on $\regexr$. See Appendix~\ref{sec:proof-firstderivative}.
\end{proof}
Motivated by this result, we define the Brzozowski derivative for a non-empty subset $\literall$ of a literal in
$\getFirst\regexr$. This definition involves an arbitrary choice of $\symbola\in\literall$, but this choice does not
influence the calculated derivative according to Lemma~\ref{lemma:partial-equivalence}, part 1.
\begin{definition}\label{def:derivative-literals}
  Let $\literall' \in \getFirst\regexr$. For each $\emptyset \ne \literall \subseteq \literall'$ define
  $\deriv\literall\regexr \df \deriv\symbola\regexr$, where $\symbola \in \literall$.
\end{definition}
\begin{lemma}[Coverage]\label{thm:coverage}
  For all $\symbola$, $\wordu$, and $\regexr$ it holds that:
  \begin{align*}
    \wordu\in\lang{\deriv{\symbola}{\regexr}}~\Leftrightarrow~\exists\literall\in\getFirst{\regexr}:
    \symbola \in \literall \wedge
    \wordu\in\lang{\pderiv{\literall}{\regexr}} \wedge
    \wordu\in\lang{\nderiv{\literall}{\regexr}}
  \end{align*}
\end{lemma}
\begin{proof}[Proof of Lemma~\ref{thm:coverage}]
  This result follows from Theorem~\ref{thm:firstderivative} and Lemma~\ref{thm:first}.
\end{proof}
We conclude that to determine a finite set of  representatives for all derivatives of a regular
expression $\regexr$ it is sufficient to select one symbol $\symbola$ from each equivalence class 
$\literall \in \getFirst\regexr \setminus \{\emptyset\}$ and calculate $\deriv\symbola\regexr$.
Alternatively, we may calculate $\pderiv\literall\regexr$ or
$\nderiv\literall\regexr$ according to Theorem~\ref{thm:firstderivative}.
It remains to lift this result to solving inequalities.

\section{Solving Inequalities}
\label{sec:inequalities}

Theorem~\ref{thm:containment-antimirov} is the foundation of
Antimirov's algorithm. It turns out that we can prove a stronger
version of this theorem, which makes the rules \RefTirName{CC-Disprove}
and \RefTirName{CC-Unfold} sound and complete and which also encompasses
the soundness of the restriction to first sets.

\begin{theorem}[Containment]\label{thm:symbol-containment}
  \begin{align*}
    \isSubSetOf{\regexr}{\regexs}~\Leftrightarrow&~(\isNullable{\regexr}\Rightarrow\isNullable{\regexs})~\wedge~(\forall\symbola\in\getNext{\regexr})~\isSubSetOf{\deriv{\symbola}{\regexr}}{\deriv{\symbola}{\regexs}}
  \end{align*}
\end{theorem}

\begin{proof}[Proof of Theorem~\ref{thm:symbol-containment}]
  See Appendix~\ref{sec:proof-semantic-containment}.
\end{proof}
As we remarked before, it may be very expensive (or even impossible) to construct all derivatives with respect to the first
symbols, particularly for negated expressions and for large or infinite alphabets. To obtain a decision procedure for
containment, we need a finite set of derivatives. Therefore, we use next literals as representatives of the first
symbols and use Brzozowski derivatives on literals (Definition~\ref{def:derivative-literals}) on both sides.

To define the next literals of an inequality $\checkSubSetOf{\regexr}{\regexs}$, it would be sound to use the join of
the next literals of both sides: $\getFirst\regexr \literaljoin \getFirst\regexs$. However, we can do slightly
better. Theorem~\ref{thm:symbol-containment} proves that the first symbols of $\regexr$ are sufficient to prove
containment. Using the full join operation, however, would cover $\getNext\regexr \cup \getNext\regexs$ (by
Lemma~\ref{thm:first}). Hence, we define a left-biased version of the join operator that only covers the symbols of its
left operand.
\begin{definition}[Left Join]\label{def:left-join}
  Let $\literalset_1$ and $\literalset_2$ be two sets of mutually
  disjoint literals.
  \begin{align*}
    \literalset_1 \literalleftjoin \literalset_2~\df&
    \{
      (\literall_1\literalcap\literall_2),
      (\literall_1\literalcap\inv{\literalbigcup\literalset_2})
      \mid
    \literall_1\in \literalset_1, \literall_2\in \literalset_2\}
  \end{align*}
\end{definition}
The following lemma states the properties of the left join operation.
\begin{lemma}[Properties of Left Join]\label{lemma:properties-of-left-join}
  Let $\literalset_1$ and $\literalset_2$ be non-empty sets of mutually disjoint literals.
  \begin{enumerate}
    \item $\bigcup (\literalset_1 \literalleftjoin \literalset_2) = \bigcup \literalset_1$.
    \item $(\forall \literall\ne\literall' \in \literalset_1 \literalleftjoin \literalset_2)$ $\literall\sqcap\literall' = \emptyset$.
    \item $(\forall \literall \in \literalset_1 \literalleftjoin \literalset_2)$ $(\forall \literall_i \in \literalset_i)$
      $\literall \sqcap \literall_i \ne \emptyset \Rightarrow \literall \sqsubseteq \literall_i$.
  \end{enumerate}
\end{lemma}
\begin{proof}[Proof of Lemma~\ref{lemma:properties-of-left-join}]
  Analogous to the proof of Lemma~\ref{lemma:properties-of-join} in Appendix~\ref{sec:properties-join}.
\end{proof}
\begin{definition}[Next Literals of an Inequality]
  Let $\checkSubSetOf{\regexr}{\regexs}$ be an inequality.
  \begin{align*}
    \getFirst{\checkSubSetOf{\regexr}{\regexs}} \df
    \getFirst\regexr \literalleftjoin \getFirst\regexs
  \end{align*}
\end{definition}
Finally, we can state a generalization of Antimirov's containment theorem for
\ERE{}s, where each unfolding step generates only finitely many derivatives.
\begin{theorem}[Containment]\label{thm:containment}
  For all regular expressions $\regexr$ and $\regexs$,
  \begin{align*}
    \isSubSetOf{\regexr}{\regexs}~\Leftrightarrow&
    ~(\isNullable{\regexr}\Rightarrow\isNullable{\regexs})
    ~\wedge~
    (\forall\literall\in\getFirst{\checkSubSetOf{\regexr}{\regexs}})~%
    \isSubSetOf{\deriv{\literall}{\regexr}}{\deriv{\literall}{\regexs}}.
  \end{align*}
\end{theorem}
\begin{proof}[Proof of Theorem~\ref{thm:containment}]
  The proof is by contraposition.
  If $\isNotSubSetOf{\regexr}{\regexs}$ then $\exists\literall\in\getFirst{\checkSubSetOf{\regexr}{\regexs}}:$
  $ \isNotSubSetOf{\deriv{\literall}{\regexr}}{\deriv{\literall}{\regexs}}$ or
  $\neg(\isNullable{\regexr}\Rightarrow\isNullable{\regexs})$.
  See also Appendix~\ref{sec:proof-symbolic-containment}.
\end{proof}
For $\literall\in\getFirst{\checkSubSetOf{\regexr}{\regexs}}$ define
$\nderiv\literall{\checkSubSetOf{\regexr}{\regexs}}
\df (\checkSubSetOf{\nderiv\literall\regexr}{\pderiv\literall\regexs})
= (\checkSubSetOf{\deriv\literall\regexr}{\deriv\literall\regexs})$.
\begin{theorem}[Finiteness]\label{thm:finiteness}
  Let $R$ be a finite set of regular inequalities. Define
  \begin{displaymath}
    F (R) =
    R \cup \{ \nderiv{\literall}{ \checkSubSetOf{{\regexr}}{{\regexs}}} \mid \checkSubSetOf{\regexr}{\regexs} \in R, \literall\in\getFirst{\checkSubSetOf{\regexr}{\regexs}} \}
  \end{displaymath}
  For each $\regexr$ and $\regexs$, the set $\bigcup_{i\in\nat}F^{(i)}
  (\{\isSubSetOf{\regexr}{\regexs}\})$ is finite.
\end{theorem}
\begin{proof}[Proof of Theorem~\ref{thm:finiteness}]
  As we consider regular expressions up to similarity (as defined by Brzozowski \cite{Brzozowski1964}) and
  $\nderiv{\literall}{ \checkSubSetOf{{\regexr}}{{\regexs}}} = \checkSubSetOf{\deriv\literall\regexr}{\deriv\literall\regexs}$ is
  essentially applying the Brzozowski derivative to a pair of (extended) regular expressions, we know that the set of
  these pairs is finite (because there are only finitely many dissimilar iterated Brzozowski derivatives for a regular
  expression \cite{Brzozowski1964}). 
\end{proof}
\begin{figure}[t]
  \begin{mathpar}
    \inferrule[\RuleCCDisprove]
    {%
      \isNullable{\regexr}\\
      \neg\isNullable{\regexs}
    }
    {%
      \ccCtx~\entails~\checkSubSetOf{\regexr}{\regexs}~:~\false%
    }
    \and
    \inferrule[\RuleCCDelete]
    {%
      \inCcCtx{\checkSubSetOf{\regexr}{\regexs}}%
    }
    {%
      \ccCtx~\entails~\checkSubSetOf{\regexr}{\regexs}~:~\true%
    }\and
    \inferrule[\RuleCCUnfoldTrue]
    {%
      \notInCcCtx{\checkSubSetOf{\regexr}{\regexs}}\\
      \isNullable\regexr \Rightarrow \isNullable\regexs\\
      \forall\literall\in\getFirst{\checkSubSetOf{\regexr}{\regexs}}:
      ~\ccCtx \cup\{\checkSubSetOf{\regexr}{\regexs}\}~\entails~
      \checkSubSetOf{\deriv{\literall}{\regexr}}{\deriv{\literall}{\regexs}}~:~\true%
    }
    {%
      \ccCtx~\entails~\checkSubSetOf{\regexr}{\regexs}~:~\true%
    }\and
    \inferrule[\RuleCCUnfoldFalse]
    {%
      \notInCcCtx{\checkSubSetOf{\regexr}{\regexs}}\\
      \isNullable\regexr \Rightarrow \isNullable\regexs\\
      \exists\literall\in\getFirst{\checkSubSetOf{\regexr}{\regexs}}:
      ~\ccCtx \cup\{\checkSubSetOf{\regexr}{\regexs}\}~\entails~
      \checkSubSetOf{\deriv{\literall}{\regexr}}{\deriv{\literall}{\regexs}}~:~\false%
    }
    {%
      \ccCtx~\entails~\checkSubSetOf{\regexr}{\regexs}~:~\false%
    }
  \end{mathpar}
  \caption{Decision procedure for containment.}\label{fig:cc_unfold}
\end{figure}
These results are the basis for a complete decision procedure for
solving inequalities on extended regular expressions where literals are defined via an effective boolean algebra.
Figure~\ref{fig:cc_unfold} defines this procedure as a judgment of the form
$\ccCtx\entails\checkSubSetOf{\regexr}{\regexs}~:~b$, where $\ccCtx$ is a set of previous visited inequalities
$\checkSubSetOf\regexr\regexs$ with $\isNullable\regexr \Rightarrow \isNullable\regexs$ that are assumed to be true and
$b \in \{\true, \false\}$. The effective boolean algebra comes into play in the computation of the next literals
and in the computation of the derivatives.

Rule \Rule{\RuleCCDisprove} detects contradictory inequalities in the same way as Antimirov's rule \RefTirName{CC-Disprove}.
Rule \Rule{\RuleCCDelete} detects circular reasoning: Under the assumption that $\checkSubSetOf\regexr\regexs$ holds we
were not (yet) able to derive a contradiction and thus conclude that $\checkSubSetOf\regexr\regexs$ holds. This rule
guarantees termination because of the finiteness result (Theorem~\ref{thm:finiteness}).
The rules \Rule{\RuleCCUnfoldTrue} and \Rule{\RuleCCUnfoldFalse} apply only if $\checkSubSetOf{\regexr}{\regexs}$ is
neither contradictory nor in the context. A deterministic implementation would generate the literals
$\literall\in\getFirst{\checkSubSetOf\regexr\regexs}$ and recursively check
$\nderiv{\literall}{\checkSubSetOf{\regexr}{\regexs}}$. If any of these checks returns false, then
\Rule{\RuleCCUnfoldFalse} fires. Otherwise \Rule{\RuleCCUnfoldTrue} signals a successful containment
proof. Theorem~\ref{thm:containment} is the basis for soundness and completeness of the unfolding rules.
\begin{theorem}[Soundness]\label{thm:soundness}
  For all regular expression $\regexr$ and $\regexs$:
  \begin{align*}
    \emptyset~\entails~\checkSubSetOf{\regexr}{\regexs}~:~\top~\Leftrightarrow~\isSubSetOf{\regexr}{\regexs}
  \end{align*}
\end{theorem}
\begin{proof}[Proof of Theorem~\ref{thm:soundness}]
  We prove that $\ccCtx \entails \checkSubSetOf{\regexr}{\regexs} ~:~ \false$ iff $\isNotSubSetOf{\regexr}{\regexs}$, for
  all contexts $\ccCtx$ where $\checkSubSetOf{\regexr}{\regexs} \notin \ccCtx$. This is sufficient because each regular
  inequality gives rise to a finite derivation by Theorem~\ref{thm:finiteness}.
  See Appendix~\ref{sec:proof-soundness} for details.
\end{proof}
\begin{figure}[t]
  \centering
  \begin{mathpar}
    \inferrule[\RuleCCProveIdentity]
    {%
    }
    {%
      \ccCtx~\entails~\isSubSetOf{\regexr}{\regexr}~:~\true%
    }\and
    \inferrule[\RuleCCProveEmpty]
    {%
    }
    {%
      \ccCtx~\entails~\isSubSetOf{\regexNull}{\regexs}~:~\true%
    }\and
    \inferrule[\RuleCCProveNullable]
    {%
      \isNullable{\regexs}
    }
    {%
      \ccCtx~\entails~\isSubSetOf{\regexEmpty}{\regexs}~:~\true%
    }\and
    \inferrule[\RuleCCDisproveEmpty]
    {%
      \exists\literall\in\getFirst{\regexr}:~{\literall}\neq\emptyset%
    }
    {%
      \ccCtx~\entails~\isSubSetOf{\regexr}{\regexNull}~:~\false%
    }
  \end{mathpar}
  \caption{Prove and disprove axioms.}\label{fig:cc_axioms}
\end{figure}
In addition to the rules from Figure~\ref{fig:cc_unfold}, we may add
auxiliary rules to detect trivially consistent or inconsistent
inequalities early (Figure~\ref{fig:cc_axioms} contains some examples). Such rules may be used to improve
efficiency. They decide containment directly instead of unfolding repeatedly. 

\section{Conclusion}

We extended Antimirov's algorithm for proving containment of regular expressions to extended regular expressions on
potentially infinite alphabets. To work effectively with such alphabets, we require that literals in regular expressions
are drawn from an effective boolean algebra. As a slight difference, we work with Brzozowski derivatives instead of
Antimirov's notion of partial derivative.

The main effort in lifting Antimirov's algorithm is to identify, for each regular inequality
$\checkSubSetOf\regexr\regexs$, a finite set of symbols such that calculating the derivation with respect to these
symbols covers all possible derivations with all symbols. We regard the construction of the set of suitable
representatives, embodied in the notion of next literals $\getFirst{\checkSubSetOf\regexr\regexs}$, as a key
contribution of this work.

\bibliography{abbrevs,abbrv,books,misc,papers,collections,theses}

\newpage
\appendix
\section{Positive and Negative Derivatives}
\label{sec:derivatives_full}

This sections shows the full definition of the \emph{positive} and \emph{negative} derivative operator. The operators are defined by induction and flip on the complement operator.

\subsection{Positive Derivatives}

For all literals $\literall\neq\emptyset$:

\begin{align*}
  \begin{array}{lll}
    \pderiv{\literall}{\regexEmpty} &\df{}& \regexNull\\       
    \pderiv{\literall}{\setb} &\df{}& \begin{cases}
      \regexEmpty, & \literall\sqcap{\setb}\ne\bot \\
      \regexNull, & \textit{otherwise}
    \end{cases}\\
    \pderiv{\literall}{\regexStar{\regexr}} &\df{}& \regexConcat{\pderiv{\literall}{\regexr}}{\regexStar{\regexr}}\\
    \pderiv{\literall}{\regexOr{\regexr}{\regexs}} &\df{}& \regexOr{\pderiv{\literall}{\regexr}}{\pderiv{\literall}{\regexs}}\\
    \pderiv{\literall}{\regexAnd{\regexr}{\regexs}} &\df{}& \regexAnd{\pderiv{\literall}{\regexr}}{\pderiv{\literall}{\regexs}}\\
    \pderiv{\literall}{\regexNeg{\regexr}} &\df{}&  \regexNeg{\nderiv{\literall}{\regexr}}\\
    \pderiv{\literall}{\regexConcat{\regexr}{\regexs}} &\df{}& \begin{cases}
      \regexOr{\regexConcat{\pderiv{\literall}{\regexr}}{\regexs}}{\pderiv{\literall}{\regexs}}, &\isNullable{\regexr}\\
      \regexConcat{\pderiv{\literall}{\regexr}}{\regexs}  , &\textit{otherwise}
    \end{cases}\\
  \end{array}
\end{align*}

\subsection{Negative Derivatives}

For all literals $\literall\neq\emptyset$:

\begin{align*}
  \begin{array}{lll}
    \nderiv{\literall}{\regexEmpty} &\df{}& \regexNull\\       
    \nderiv{\literall}{\setb} &\df{}& \begin{cases}
      \regexEmpty, & {\literall}\sqcap\overline{\setb}=\bot\\
      \regexNull, & \textit{otherwise}
    \end{cases}\\
    \nderiv{\literall}{\regexStar{\regexr}} &\df{}& \regexConcat{\nderiv{\literall}{\regexr}}{\regexStar{\regexr}}\\
    \nderiv{\literall}{\regexOr{\regexr}{\regexs}} &\df{}& \regexOr{\nderiv{\literall}{\regexr}}{\nderiv{\literall}{\regexs}}\\
    \nderiv{\literall}{\regexAnd{\regexr}{\regexs}} &\df{}& \regexAnd{\nderiv{\literall}{\regexr}}{\nderiv{\literall}{\regexs}}\\
    \nderiv{\literall}{\regexNeg{\regexr}} &\df{}&  \regexNeg{\pderiv{\literall}{\regexr}}\\
    \nderiv{\literall}{\regexConcat{\regexr}{\regexs}} &\df{}& \begin{cases}
      \regexOr{\regexConcat{\nderiv{\literall}{\regexr}}{\regexs}}{\nderiv{\literall}{\regexs}}, &\isNullable{\regexr}\\
      \regexConcat{\nderiv{\literall}{\regexr}}{\regexs}  , &\textit{otherwise}
    \end{cases}\\
  \end{array}
\end{align*}

\newpage
\section{Complexity}
\label{sec:complexity}

This section comprises the complexity of the decision procedure. The complexity has two sources: building the next literals and computing the derivatives. 

We express the complexity in terms of the size of a regular expression. The size is directly related to the number of derivation steps and to the number of operations if gathering the next literals.

\begin{definition}[Size]\label{def:size}
  The size $\sizeOf{\regexr}$ of a regular expression $\regexr$ is the number of expression constructors and literals.
  \begin{mathpar}
    \begin{array}{lll}
      \sizeOf{\regexEmpty} &=& 1 \\
      \sizeOf{\regexSet} &=& 1\\
      ~\\
    \end{array}

    \begin{array}{lll}
      \sizeOf{\regexStar{\regexr}} &=& \sizeOf{\regexr} + 1\\
      \sizeOf{\regexOr{\regexr}{\regexs}} &=& \sizeOf{\regexr} + \sizeOf{\regexs} + 1\\
      \sizeOf{\regexConcat{\regexr}{\regexs}} &=& \sizeOf{\regexr} + \sizeOf{\regexs} + 1
    \end{array}

    \begin{array}{lll}
      \sizeOf{\regexAnd{\regexr}{\regexs}} &=& \sizeOf{\regexr} + \sizeOf{\regexs} + 1\\
      \sizeOf{\regexNeg{\regexr}} &=& \sizeOf{\regexr} + 1\\
      ~\\
    \end{array}
  \end{mathpar}
\end{definition}
The number of literals in a regular expression is another useful
measure. 
\begin{definition}[Literal Width]\label{def:width}
  The literal width $\widthOf{\regexr}$ of a regular expression $\regexr$ denotes the total number of literals $\literall$ in $\regexr$.
\end{definition}
Calculating the next literals for $\regexr$ may require a
number of operations on the symbol set representation which is
exponential in the literal width $\widthOf\regexr$, because there are
regular expressions where the number of next literals is already
exponential. For example, consider $\regexr =
\regexAnd{\regexAnd{(\regexOr{\seta_1}{\setb_1})}{(\regexOr{\seta_2}{\setb_2})}}{\dots(\regexOr{\seta_n}{\setb_n})}$
with $\widthOf\regexr = 2n$. With a sufficiently large alphabet, we may choose the sets
$\seta_i$ and $\setb_i$ such that $|\getFirst\regexr| = 2^n$.

The number of different derivatives of a regular expression
is bounded by $2^{\sizeOf\regexr}$
analogously to Brzozowski's result. Hence, the number of different
derivatives of a regular expression inequality $\isSubSetOf\regexr\regexs$ is bounded by
$2^{\sizeOf\regexr + \sizeOf\regexs}$. Taken together, our decision
procedure requires the computation of an exponential number of
derivative operations and, for the result of each of these operations,
a new set of next literals has to be determined, in the worst case. 

The derivative itself runs in constant time in most cases. However, in
the case where the argument expression is a symbol-set literal, a
calculation on the representation of symbol sets is required.

\newpage
\section{Lemma~\ref{thm:derivatives}: Positive and Negative Derivatives}
\label{sec:proof-derivative}

\begin{proof}[Proof of Lemma~\ref{thm:derivatives}]
  For any \ERE{} $\regexr$ and for any literal $\literall$, the following equation holds:
  \begin{gather}
    \lang{\nderiv{\literall}{\regexr}} ~\subseteq~ \bigcap_{\symbola\in{\literall}} \lang{\deriv{\symbola}{\regexr}}
  \end{gather}
  \begin{gather}
    \lang{\pderiv{\literall}{\regexr}} ~\supseteq~ \bigcup_{\symbola\in{\literall}} \lang{\deriv{\symbola}{\regexr}}
  \end{gather}
  Proof by induction on $\regexr$.

  \begin{description}

    \item[Case] $\regexr=\regexEmpty$:~ Claim holds because $\lang{\nderiv{\literall}{\regexEmpty}} = \lang{\pderiv{\literall}{\regexEmpty}} = \regexNull$.

    \item[Case] $\regexr=\setb$:~ Claim holds because 
      \begin{align}
        \lang{\nderiv{\seta}{\setb}} ~=~ \bigcap_{\symbola\in\seta} \lang{\deriv{\symbola}{\setb}} ~&=~ \begin{cases}
          \{\regexEmpty\}, & \seta\subseteq\setb\\
          \emptyset, & otherwise
        \end{cases}
      \end{align}
      and 
      \begin{align}
        \lang{\pderiv{\seta}{\setb}} ~=~ \bigcup_{\symbola\in\seta} \lang{\deriv{\symbola}{\setb}} ~&=~ \begin{cases}
          \{\regexEmpty\}, & \seta\cap\setb\neq\emptyset\\
          \emptyset, & otherwise
        \end{cases}
      \end{align}

    \item[Case] $\regexr=\regexStar{\regexs}$:~ By induction 
      \begin{align}
        &\lang{\nderiv{\literall}{\regexs}} ~\subseteqIH~ \bigcap_{\symbola\in{\literall}} \lang{\deriv{\symbola}{\regexs}}\\
        &\lang{\pderiv{\literall}{\regexs}} ~\supseteqIH~ \bigcup_{\symbola\in{\literall}} \lang{\deriv{\symbola}{\regexs}}
      \end{align}
      holds.
      We obtain that
      \begin{align}
        &\forall\symbola:~ \lang{\deriv{\symbola}{\regexStar{\regexs}}} ~=~ \lang{\deriv{\symbola}{\regexs}\concat\regexStar{\regexs}}\\
        &\forall\literall:~ \lang{\nderiv{\literall}{\regexStar{\regexs}}} ~=~ \lang{\nderiv{\literall}{\regexs}\concat\regexStar{\regexs}}\\
        &\forall\literall:~ \lang{\pderiv{\literall}{\regexStar{\regexs}}} ~=~ \lang{\pderiv{\literall}{\regexs}\concat\regexStar{\regexs}}
      \end{align}
      holds. Claim holds because
      \begin{align}
        \forall\literall:~ \lang{\nderiv{\literall}{\regexStar{\regexs}}}
        ~=&~ \lang{\nderiv{\literall}{\regexs}\concat\regexStar{\regexs}}\\
        ~\subseteqIH&~ \{ \wordu\wordv ~|~ \wordu\in\bigcap_{\symbola\in{\literall}} \lang{\deriv{\symbola}{\regexs}}, \wordv\in\lang{\regexStar{\regexs}}\} \\
        ~=&~ \bigcap_{\symbola\in{\literall}} \{ \wordu\wordv ~|~ \wordu\in\lang{\deriv{\symbola}{\regexs}}, \wordv\in\lang{\regexStar{\regexs}}\} \\
        ~=&~ \bigcap_{\symbola\in{\literall}} \lang{\deriv{\symbola}{\regexStar{\regexs}}}
      \end{align}
      and 
      \begin{align}
        \forall\literall:~ \lang{\pderiv{\literall}{\regexStar{\regexs}}}
        ~=&~ \lang{\pderiv{\literall}{\regexs}\concat\regexStar{\regexs}}\\
        ~\supseteqIH&~ \{ \wordu\wordv ~|~ \wordu\in\bigcup_{\symbola\in{\literall}} \lang{\deriv{\symbola}{\regexs}}, \wordv\in\lang{\regexStar{\regexs}}\} \\
        ~=&~ \bigcup_{\symbola\in{\literall}} \{ \wordu\wordv ~|~ \wordu\in\lang{\deriv{\symbola}{\regexs}}, \wordv\in\lang{\regexStar{\regexs}}\} \\
        ~=&~ \bigcup_{\symbola\in{\literall}} \lang{\deriv{\symbola}{\regexStar{\regexs}}}
      \end{align}

    \item[Case] $\regexr=\regexOr{\regexs}{\regext}$:~ By induction 
      \begin{align}
        &\lang{\nderiv{\literall}{\regexs}} ~\subseteqIH~ \bigcap_{\symbola\in{\literall}} \lang{\deriv{\symbola}{\regexs}}\\
        &\lang{\nderiv{\literall}{\regext}} ~\subseteqIH~ \bigcap_{\symbola\in{\literall}} \lang{\deriv{\symbola}{\regext}}\\
        &\lang{\pderiv{\literall}{\regexs}} ~\supseteqIH~ \bigcup_{\symbola\in{\literall}} \lang{\deriv{\symbola}{\regexs}}\\
        &\lang{\pderiv{\literall}{\regext}} ~\supseteqIH~ \bigcup_{\symbola\in{\literall}} \lang{\deriv{\symbola}{\regext}}
      \end{align}
      holds.
      We obtain that
      \begin{align}
        &\lang{\deriv{\symbola}{\regexOr{\regexs}{\regext}}} ~=~ \lang{\deriv{\symbola}{\regexs}}\cup\lang{\deriv{\symbola}{\regext}}\\
        &\lang{\nderiv{\literall}{\regexOr{\regexs}{\regext}}} ~=~ \lang{\nderiv{\literall}{\regexs}}\cup\lang{\nderiv{\literall}{\regext}}\\
        &\lang{\pderiv{\literall}{\regexOr{\regexs}{\regext}}} ~=~ \lang{\pderiv{\literall}{\regexs}}\cup\lang{\pderiv{\literall}{\regext}}
      \end{align}
      holds. Claim holds because
      \begin{align}
        \lang{\nderiv{\literall}{\regexOr{\regexs}{\regext}}}
        ~=&~ \lang{\nderiv{\literall}{\regexs}}\cup\lang{\nderiv{\literall}{\regext}}\\
        ~\subseteqIH&~ \bigcap_{\symbola\in{\literall}} \lang{\deriv{\symbola}{\regexs}} \cup \bigcap_{\symbola\in{\literall}} \lang{\deriv{\symbola}{\regext}}\\
        ~\subseteq&~ \bigcap_{\symbola\in{\literall}}  \lang{\deriv{\symbola}{\regexs}} \cup \lang{\deriv{\symbola}{\regext}}\\
        ~=&~ \bigcap_{\symbola\in{\literall}}  \lang{\deriv{\symbola}{\regexOr{\regexs}{\regext}}}                      
      \end{align}
      and
      \begin{align}
        \lang{\pderiv{\literall}{\regexOr{\regexs}{\regext}}}
        ~=&~ \lang{\pderiv{\literall}{\regexs}}\cup\lang{\pderiv{\literall}{\regext}}\\
        ~\supseteqIH&~ \bigcup_{\symbola\in{\literall}} \lang{\deriv{\symbola}{\regexs}} \cup \bigcup_{\symbola\in{\literall}} \lang{\deriv{\symbola}{\regext}}\\
        ~=&~ \bigcup_{\symbola\in{\literall}}  \lang{\deriv{\symbola}{\regexs}} \cup \lang{\deriv{\symbola}{\regext}}\\
        ~=&~ \bigcup_{\symbola\in{\literall}}  \lang{\deriv{\symbola}{\regexOr{\regexs}{\regext}}}                      
      \end{align}

    \item[Case] $\regexr=\regexAnd{\regexs}{\regext}$:~ By induction
      \begin{align}
        &\lang{\nderiv{\literall}{\regexs}} ~\subseteqIH~ \bigcap_{\symbola\in{\literall}} \lang{\deriv{\symbola}{\regexs}}\\
        &\lang{\nderiv{\literall}{\regext}} ~\subseteqIH~ \bigcap_{\symbola\in{\literall}} \lang{\deriv{\symbola}{\regext}}\\
        &\lang{\pderiv{\literall}{\regexs}} ~\supseteqIH~ \bigcup_{\symbola\in{\literall}} \lang{\deriv{\symbola}{\regexs}}\\
        &\lang{\pderiv{\literall}{\regext}} ~\supseteqIH~ \bigcup_{\symbola\in{\literall}} \lang{\deriv{\symbola}{\regext}}
      \end{align}
      holds.
      We obtain that
      \begin{align}
        &\lang{\deriv{\symbola}{\regexAnd{\regexs}{\regext}}} ~=~ \lang{\deriv{\symbola}{\regexs}}\cap\lang{\deriv{\symbola}{\regext}}\\
        &\lang{\nderiv{\literall}{\regexAnd{\regexs}{\regext}}} ~=~ \lang{\nderiv{\literall}{\regexs}}\cap\lang{\nderiv{\literall}{\regext}}\\
        &\lang{\pderiv{\literall}{\regexAnd{\regexs}{\regext}}} ~=~ \lang{\pderiv{\literall}{\regexs}}\cap\lang{\pderiv{\literall}{\regext}}
      \end{align}
      holds. Claim holds because
      \begin{align}
        \lang{\nderiv{\literall}{\regexAnd{\regexs}{\regext}}}
        ~=&~ \lang{\nderiv{\literall}{\regexs}}\cap\lang{\nderiv{\literall}{\regext}}\\
        ~\subseteqIH&~ \bigcap_{\symbola\in{\literall}} \lang{\deriv{\symbola}{\regexs}} \cap \bigcap_{\symbola\in{\literall}} \lang{\deriv{\symbola}{\regext}}\\
        ~=&~ \bigcap_{\symbola\in{\literall}}  \lang{\deriv{\symbola}{\regexs}} \cap \lang{\deriv{\symbola}{\regext}}\\
        ~=&~ \bigcap_{\symbola\in{\literall}}  \lang{\deriv{\symbola}{\regexAnd{\regexs}{\regext}}}                     
      \end{align}
      and
      \begin{align}
        \lang{\pderiv{\literall}{\regexAnd{\regexs}{\regext}}}
        ~=&~ \lang{\pderiv{\literall}{\regexs}}\cap\lang{\pderiv{\literall}{\regext}}\\
        ~\supseteqIH&~ \bigcup_{\symbola\in{\literall}} \lang{\deriv{\symbola}{\regexs}} \cap \bigcup_{\symbola\in{\literall}} \lang{\deriv{\symbola}{\regext}}\\
        ~\supseteq&~ \bigcup_{\symbola\in{\literall}}  \lang{\deriv{\symbola}{\regexs}} \cap \lang{\deriv{\symbola}{\regext}}\\
        ~=&~ \bigcup_{\symbola\in{\literall}}  \lang{\deriv{\symbola}{\regexAnd{\regexs}{\regext}}}                     
      \end{align}

    \item[Case] $\regexr=\regexNeg{\regexs}$:~ By induction 
      \begin{align}
        &\lang{\nderiv{\literall}{\regexs}} ~\subseteqIH~ \bigcap_{\symbola\in{\literall}} \lang{\deriv{\symbola}{\regexs}}\\
        &\lang{\pderiv{\literall}{\regexs}} ~\supseteqIH~ \bigcup_{\symbola\in{\literall}} \lang{\deriv{\symbola}{\regexs}}
      \end{align}
      holds.
      We obtain that
      \begin{align}
        &\forall\symbola:~ \lang{\deriv{\symbola}{\regexNeg{\regexs}}} ~=~ \regexWords\backslash\lang{\deriv{\symbola}{\regexs}}\\
        &\forall\literall:~ \lang{\nderiv{\literall}{\regexNeg{\regexs}}} ~=~ \regexWords\backslash\lang{\pderiv{\literall}{\regexs}}\\
        &\forall\literall:~ \lang{\pderiv{\literall}{\regexNeg{\regexs}}} ~=~ \regexWords\backslash\lang{\nderiv{\literall}{\regexs}}
      \end{align}
      holds. Claim holds because
      \begin{align}
        \forall\literall:~ \lang{\nderiv{\literall}{\regexNeg{\regexs}}}
        ~=&~ \regexWords\backslash\lang{\pderiv{\literall}{\regexs}}\\
        ~\subseteqIH&~ \regexWords\backslash\bigcup_{\symbola\in{\literall}} \lang{\deriv{\symbola}{\regexs}}\\
        ~=&~ \bigcap_{\symbola\in{\literall}} \regexWords\backslash\lang{\deriv{\symbola}{\regexs}}\\
        ~=&~ \bigcap_{\symbola\in{\literall}} \lang{\deriv{\symbola}{\regexNeg{\regexs}}}
      \end{align}
      and 
      \begin{align}
        \forall\literall:~ \lang{\pderiv{\literall}{\regexNeg{\regexs}}}
        ~=&~ \regexWords\backslash\lang{\nderiv{\literall}{\regexs}}\\
        ~\supseteqIH&~ \regexWords\backslash\bigcap_{\symbola\in{\literall}} \lang{\deriv{\symbola}{\regexs}}\\
        ~=&~ \bigcup_{\symbola\in{\literall}} \regexWords\backslash\lang{\deriv{\symbola}{\regexs}}\\
        ~=&~ \bigcup_{\symbola\in{\literall}} \lang{\deriv{\symbola}{\regexNeg{\regexs}}}
      \end{align}

    \item[Case] $\regexr=\regexConcat{\regexs}{\regext}$:~ By induction
      \begin{align}
        &\lang{\nderiv{\literall}{\regexs}} ~\subseteqIH~ \bigcap_{\symbola\in{\literall}} \lang{\deriv{\symbola}{\regexs}}\\
        &\lang{\nderiv{\literall}{\regext}} ~\subseteqIH~ \bigcap_{\symbola\in{\literall}} \lang{\deriv{\symbola}{\regext}}\\
        &\lang{\pderiv{\literall}{\regexs}} ~\supseteqIH~ \bigcup_{\symbola\in{\literall}} \lang{\deriv{\symbola}{\regexs}}\\
        &\lang{\pderiv{\literall}{\regext}} ~\supseteqIH~ \bigcup_{\symbola\in{\literall}} \lang{\deriv{\symbola}{\regext}}
      \end{align}
      holds.
      We obtain that 
      \begin{align}
        &\forall\symbola:~ \lang{\deriv{\symbola}{\regexConcat{\regexs}{\regext}}} ~=~ \begin{cases}
          \lang{\deriv{\symbola}{\regexs}\concat\regext} \cup \lang{\deriv{\symbola}{\regext}}, & \isNullable{\regexs}\\
          \lang{\deriv{\symbola}{\regexs}\concat\regext}, & otherwise
        \end{cases}\\
        &\forall\literall:~ \lang{\nderiv{\literall}{\regexConcat{\regexs}{\regext}}} ~=~ \begin{cases}
          \lang{\nderiv{\literall}{\regexs}\concat\regext} \cup \lang{\nderiv{\literall}{\regext}}, & \isNullable{\regexs}\\
          \lang{\nderiv{\literall}{\regexs}\concat\regext}, & otherwise
        \end{cases}\\
        &\forall\literall:~ \lang{\pderiv{\literall}{\regexConcat{\regexs}{\regext}}} ~=~ \begin{cases}
          \lang{\pderiv{\literall}{\regexs}\concat\regext} \cup \lang{\pderiv{\literall}{\regext}}, & \isNullable{\regexs}\\
          \lang{\pderiv{\literall}{\regexs}\concat\regext}, & otherwise
        \end{cases}
      \end{align}
      holds.
      \begin{description}
        \item[Subcase] $\isNullable{\regexs}$:~
          Claim holds because
          \begin{align}
            \lang{\nderiv{\literall}{\regexConcat{\regexs}{\regext}}}
            ~=&~ \lang{\nderiv{\literall}{\regexs}\concat\regext} \cup \lang{\nderiv{\literall}{\regext}}\\
            ~\subseteqIH&~ \{ \wordu\wordv ~|~ \wordu\in\bigcap_{\symbola\in{\literall}} \lang{\deriv{\symbola}{\regexs}}, \wordv\in\lang{\regext}\} \cup \bigcap_{\symbola\in{\literall}} \lang{\deriv{\symbola}{\regext}}\\
            ~\subseteq&~ \bigcap_{\symbola\in{\literall}}\{ \wordu\wordv ~|~ \lang{\deriv{\symbola}{\regexs}}, \wordv\in\lang{\regext}\} \cup \lang{\deriv{\symbola}{\regext}}\\
            ~=&~ \bigcap_{\symbola\in{\literall}} \lang{\deriv{\symbola}{\regexs}\concat\regext} \cup \lang{\deriv{\symbola}{\regext}}\\
            ~=&~ \bigcap_{\symbola\in{\literall}} \deriv{\symbola}{\regexs\concat\regext}
          \end{align}
          and
          \begin{align}
            \lang{\pderiv{\literall}{\regexConcat{\regexs}{\regext}}}
            ~=&~ \lang{\pderiv{\literall}{\regexs}\concat\regext} \cup \lang{\pderiv{\literall}{\regext}}\\
            ~\supseteqIH&~ \{ \wordu\wordv ~|~ \wordu\in\bigcup_{\symbola\in{\literall}} \lang{\deriv{\symbola}{\regexs}}, \wordv\in\lang{\regext}\} \cup \bigcup_{\symbola\in{\literall}} \lang{\deriv{\symbola}{\regext}}\\
            ~=&~ \bigcup_{\symbola\in{\literall}}\{ \wordu\wordv ~|~ \lang{\deriv{\symbola}{\regexs}}, \wordv\in\lang{\regext}\} \cup \lang{\deriv{\symbola}{\regext}}\\
            ~=&~ \bigcup_{\symbola\in{\literall}} \lang{\deriv{\symbola}{\regexs}\concat\regext} \cup \lang{\deriv{\symbola}{\regext}}\\
            ~=&~ \bigcup_{\symbola\in{\literall}} \deriv{\symbola}{\regexs\concat\regext}
          \end{align}
        \item[Subcase] $\neg\isNullable{\regexs}$:~
          Claim holds because
          \begin{align}
            \lang{\nderiv{\literall}{\regexConcat{\regexs}{\regext}}}
            ~=&~ \lang{\nderiv{\literall}{\regexs}\concat\regext}\\
            ~\subseteqIH&~ \{ \wordu\wordv ~|~ \wordu\in\bigcap_{\symbola\in{\literall}} \lang{\deriv{\symbola}{\regexs}}, \wordv\in\lang{\regext}\}\\
            ~=&~ \bigcap_{\symbola\in{\literall}}\{ \wordu\wordv ~|~ \lang{\deriv{\symbola}{\regexs}}, \wordv\in\lang{\regext}\}\\
            ~=&~ \bigcap_{\symbola\in{\literall}} \lang{\deriv{\symbola}{\regexs}\concat\regext}\\
            ~=&~ \bigcap_{\symbola\in{\literall}} \deriv{\symbola}{\regexs\concat\regext}
          \end{align}
          and
          \begin{align}
            \lang{\pderiv{\literall}{\regexConcat{\regexs}{\regext}}}
            ~=&~ \lang{\pderiv{\literall}{\regexs}\concat\regext}\\
            ~\supseteqIH&~ \{ \wordu\wordv ~|~ \wordu\in\bigcup_{\symbola\in{\literall}} \lang{\deriv{\symbola}{\regexs}}, \wordv\in\lang{\regext}\}\\
            ~=&~ \bigcup_{\symbola\in{\literall}}\{ \wordu\wordv ~|~ \lang{\deriv{\symbola}{\regexs}}, \wordv\in\lang{\regext}\}\\
            ~=&~ \bigcup_{\symbola\in{\literall}} \lang{\deriv{\symbola}{\regexs}\concat\regext}\\
            ~=&~ \bigcup_{\symbola\in{\literall}} \deriv{\symbola}{\regexs\concat\regext}
          \end{align}

      \end{description}
  \end{description}
\end{proof}

\newpage
\section{Lemma~\ref{lemma:properties-of-join}: Properties of Join}
\label{sec:properties-join}

Let $\literalset_1$ and $\literalset_2$ be non-empty sets of mutually disjoint literals.
\begin{enumerate}
  \item $\bigcup (\literalset_1 \literaljoin \literalset_2) = \bigcup \literalset_1 \cup \bigcup \literalset_2$.
  \item $(\forall \literall\ne\literall' \in \literalset_1 \literaljoin \literalset_2)$ $\literall\sqcap\literall' = \emptyset$.
  \item $(\forall \literall \in \literalset_1 \literaljoin \literalset_2)$ $(\forall \literall_i \in \literalset_i)$
    $\literall \sqcap \literall_i \ne \emptyset \Rightarrow \literall \sqsubseteq \literall_i$.
\end{enumerate}

\begin{proof}[Proof of Lemma~\ref{lemma:properties-of-join}]
  ~
  \begin{enumerate}
    \item Inclusion from left to right ``$\subseteq$'':
      Suppose that $\symbola \in \bigcup (\literalset_1 \literaljoin \literalset_2)$.
      Then there exists some $\literall_1\in \literalset_1$ and $\literall_2\in \literalset_2$ such that
      \begin{itemize}
        \item 
          $\symbola \in \literall_1\literalcap\literall_2$, but then
          $\symbola\in \literall_1 \subseteq \bigcup \literalset_1 \cup \bigcup \literalset_2$;
        \item
          $\symbola \in \literall_1\literalcap\inv{\literalbigcup\literalset_2}$, but then 
          $\symbola\in \literall_1 \subseteq \bigcup \literalset_1 \cup \bigcup \literalset_2$; or
        \item
          $\symbola \in \inv{\literalbigcup\literalset_1}\literalcap\literall_2$, but then
          $\symbola\in \literall_2 \subseteq \bigcup \literalset_1 \cup \bigcup \literalset_2$.      
      \end{itemize}
      Inclusion from right to left ``$\supseteq$'':
      Suppose that $\symbola \in \bigcup \literalset_1 \cup \bigcup \literalset_2$.
      There are three cases.
      \begin{itemize}
        \item If there are  $\literall_1 \in \literalset_1$ such that $\symbola \in \literall_1$ and   $\literall_2 \in
          \literalset_2$ such that $\symbola \in \literall_2$, then $\symbola \in \literall_1 \literalcap \literall_2 \in
          \literalset_1 \literaljoin \literalset_2$.
        \item If there is some  $\literall_1 \in \literalset_1$ such that $\symbola \in \literall_1$ but there is no  $\literall_2 \in
          \literalset_2$ such that $\symbola \in \literall_2$, then $\symbola \in
          \literall_1\literalcap\inv{\literalbigcup\literalset_2} \in \literalset_1 \literaljoin \literalset_2$.
        \item Symmetric to previous case (exchange indices $1$ and $2$).
      \end{itemize}

    \item Suppose that
      $\literall_{01} = \inv{\literalbigcup\literalset_1}$ and
      $\literall_{02} = \inv{\literalbigcup\literalset_2}$. Clearly, $\literall_{01}$ is disjoint to any element of
      $\literalset_1$ and $\literall_{02}$ is disjoint to any element of
      $\literalset_2$. There are nine possible cases for $\literall$ and
      $\literall'$.
      To construct arbitrary elements of $\literalset_1\literaljoin\literalset_2$, we pick some $\literall_1,
      \literall_1' \in \literalset_1$ and $\literall_2, \literall_2' \in \literalset_2$. 
      \begin{itemize}
        \item $\literall = \literall_1 \literalcap \literall_2$ and $\literall' = \literall_1' \literalcap
          \literall_2'$. If $\literall \ne \literall'$, then $(\literall_1 , \literall_2) \ne (\literall_1', \literall_2')$
          and the claim follows from disjointness of $\literalset_1$ and $\literalset_2$.
        \item $\literall = \literall_1 \literalcap \literall_2$ and $\literall' = \literall_1' \literalcap \literall_{02}$.
          The claim follows from $\literall_2 \literalcap \literall_{02} = \emptyset$.
        \item $\literall = \literall_1 \literalcap \literall_2$ and $\literall' = \literall_{01} \literalcap \literall_2'$.
          The claim follows from $\literall_1 \literalcap \literall_{01} = \emptyset$.
        \item $\literall = \literall_1 \literalcap \literall_{02}$ and $\literall' = \literall_1' \literalcap
          \literall_2'$.
          The claim follows from $\literall_2' \literalcap \literall_{02} = \emptyset$.
        \item $\literall = \literall_1 \literalcap \literall_{02}$ and $\literall' = \literall_1' \literalcap
          \literall_{02}$.
          If $\literall \ne \literall'$, then $\literall_1 \ne \literall_1'$ and the claim follows from disjointness of
          $\literalset_1$.  
        \item $\literall = \literall_1 \literalcap \literall_{02}$ and $\literall' = \literall_{01} \literalcap
          \literall_2'$.
          The claim follows from $\literall_2' \literalcap \literall_{02} = \emptyset$.
        \item $\literall = \literall_{01} \literalcap \literall_2$ and $\literall' = \literall_1' \literalcap
          \literall_2'$.
          The claim follows from $\literall_{01} \literalcap \literall_1' = \emptyset$.
        \item $\literall = \literall_{01} \literalcap \literall_2$ and $\literall' = \literall_1' \literalcap
          \literall_{02}$.
          The claim follows from $\literall_{01} \literalcap \literall_1' = \emptyset$.
        \item $\literall = \literall_{01} \literalcap \literall_2$ and $\literall' = \literall_{01} \literalcap
          \literall_2'$.
          If $\literall \ne \literall'$, then $\literall_2 \ne \literall_2'$ and the claim follows from disjointness of
          $\literalset_2$.  
      \end{itemize}

    \item Immediate from the definition.
  \end{enumerate}
\end{proof}

\newpage
\section{Lemma~\ref{lemma:partial-equivalence}: Partial Equivalence}
\label{sec:partial-equivalence}

Let $\literalset = \getFirst\regexr$.
\begin{enumerate}
  \item\label{item:1} $(\forall \literall\in\literalset)$ $(\forall \symbola,\symbolb\in \literall)$
    $\deriv\symbola\regexr = \deriv\symbolb\regexr$
  \item $(\forall \symbola\notin \bigcup \literalset)$
    $\deriv\symbola\regexr \sqsubseteq \regexNull$
\end{enumerate}

\begin{proof}[Proof of Lemma~\ref{lemma:partial-equivalence}]
  We write $\symbola \sim_\literalset \symbolb$ if there exists some
  $\literall\in\literalset$ such that $\{\symbola, \symbolb\}
  \subseteq \literall$. The proof is by induction on $\regexr$.
  The equality in item~\ref{item:1} has to be read as semantic
  equality. It is not necessarily syntactic.

  \textbf{Cases} $\regexEmpty$: trivial.

  \textbf{Case} $\seta$: In this case, $\literalset = \{ \seta
  \}$. By definition of the derivative: For each $\symbola\in\seta$, $\deriv\symbola\seta =
  \regexEmpty$. For each $\symbolb\notin\seta$, $\deriv\symbolb\seta =
  \regexNull$.

  \textbf{Case} $\regexOr{\regexr}{\regexs}$: Let $\literalset_\regexr = \getFirst\regexr$,
  $\literalset_\regexs = \getFirst\regexs$, $\literalset = \literalset_\regexr \literaljoin
  \literalset_\regexs$, and $\literall \in \literalset$.

  There are three cases.
  If there exist $\literall_\regexr \in \literalset_\regexr$ and $\literall_\regexs \in
  \literalset_\regexs$ such that $\literall \sqsubseteq \literall_\regexr$ and
  $\literall \sqsubseteq \literall_\regexs$, then
  for all $\symbola, \symbolb \in \literall$ it holds that $\symbola \sim_{\literalset_\regexr}
  \symbolb$ and $\symbola \sim_{\literalset_\regexs}
  \symbolb$ such that, by induction, $\deriv\symbola\regexr = \deriv\symbolb\regexr$ and
  $\deriv\symbola\regexs = \deriv\symbolb\regexs$. Hence, $\deriv\symbola{\regexOr\regexr\regexs} =
  \deriv\symbolb{\regexOr\regexr\regexs}$ by definition of the derivative.

  If there exist $\literall_\regexr \in \literalset_\regexr$ such that $\literall \sqsubseteq
  \literall_\regexr$, but for all $\literall_\regexs \in \literalset_\regexs$
  it is the case that 
  $\literall \not\sqsubseteq \literall_\regexs$, then
  for all $\symbola, \symbolb \in \literall$ it holds that $\symbola \sim_{\literalset_\regexr}
  \symbolb$ such that, by induction, $\deriv\symbola\regexr = \deriv\symbolb\regexr$  and
  $\deriv\symbola\regexs \sqsubseteq \regexNull$ and $\deriv\symbolb\regexs \sqsubseteq \regexNull$. Hence,
  $\deriv\symbola{\regexOr\regexr\regexs} = \deriv\symbolb{\regexOr\regexr\regexs}$ by definition of the derivative.

  If there exist $\literall_\regexs \in \literalset_\regexs$ such that $\literall \sqsubseteq
  \literall_\regexs$, but for all $\literall_\regexr \in \literalset_\regexr$
  it is the case that 
  $\literall \not\sqsubseteq \literall_\regexr$, then
  for all $\symbola, \symbolb \in \literall$ it holds that $\symbola \sim_{\literalset_\regexs}
  \symbolb$ such that, by induction, $\deriv\symbola\regexs = \deriv\symbolb\regexs$  and
  $\deriv\symbola\regexr \sqsubseteq \regexNull$ and $\deriv\symbolb\regexr \sqsubseteq \regexNull$. Hence,
  $\deriv\symbola{\regexOr\regexr\regexs} = \deriv\symbolb{\regexOr\regexr\regexs}$ by definition of the derivative.

  \textbf{Case} $\regexConcat\regexr\regexs$: Similar.

  \textbf{Case} $\regexStar\regexr$: Let $\literalset = \getFirst\regexr$ and $\symbola
  \sim_\literalset \symbolb$.
  Now $\deriv\symbola{\regexStar\regexr}
  = \regexConcat{\deriv\symbola\regexr}{\regexStar\regexr}
  = \regexConcat{\deriv\symbolb\regexr}{\regexStar\regexr}
  = \deriv\symbolb{\regexStar\regexr}
  $ where the middle equality holds by induction.

  If $\symbola\notin\bigcup\literalset$, then $\deriv\symbola\regexr \sqsubseteq \regexNull$.
  Hence $\deriv\symbola{\regexStar\regexr}
  = \regexConcat{\deriv\symbola\regexr}{\regexStar\regexr}
  \sqsubseteq \regexConcat{\regexNull}{\regexStar\regexr}
  \sqsubseteq \regexNull
  $.

  \textbf{Case} $\regexAnd\regexr\regexs$:
  Let $\literalset_\regexr = \getFirst\regexr$,
  $\literalset_\regexs = \getFirst\regexs$,
  $\literalset = \literalset_\regexr \literalcap
  \literalset_\regexs$, and $\literall \in \literalset$.

  By construction of $\literalset$, there exist $\literall_\regexr \in \literalset_\regexr$ and $\literall_\regexs \in
  \literalset_\regexs$ such that $\literall \sqsubseteq \literall_\regexr$ and
  $\literall \sqsubseteq \literall_\regexs$. Thus,
  for all $\symbola, \symbolb \in \literall$ it holds that $\symbola \sim_{\literalset_\regexr}
  \symbolb$ and $\symbola \sim_{\literalset_\regexs}
  \symbolb$ such that, by induction, $\deriv\symbola\regexr = \deriv\symbolb\regexr$ and
  $\deriv\symbola\regexs = \deriv\symbolb\regexs$. Hence, $\deriv\symbola{\regexAnd\regexr\regexs} =
  \deriv\symbolb{\regexAnd\regexr\regexs}$ by definition of the derivative.

  If $\symbola\notin\bigcup\literalset$, then assume that $\symbola\notin\bigcup\literalset_\regexr$
  (the case for $\regexs$ is symmetric). By induction, $\deriv\symbola\regexr \sqsubseteq
  \regexNull$ so that $\deriv\symbola{\regexAnd\regexr\regexs} =
  \regexAnd{\deriv\symbola\regexr}{\deriv\symbola\regexs} \sqsubseteq
  \regexAnd{\regexNull}{\deriv\symbola\regexs} \sqsubseteq
  \regexNull$.

  \textbf{Case} $\regexNeg\regexr$: Let $\literalset_\regexr = \getFirst\regexr$ so that
  $\literalset = \getFirst{\regexNeg\regexr} = \literalset_\regexr \cup \{ \literalbigcap\{\inv{\literall} \mid \literall\in \literalset_\regexr\} \}$.
  Clearly, $\bigcup \literalset = \regexAlphabet$.

  If $\symbola \sim_\literalset \symbolb$, then there are two cases.
  If $\symbola \sim_{\literalset_\regexr} \symbolb$, then
  $\deriv\symbola{\regexNeg\regexr} = \regexNeg{\deriv\symbola\regexr} =
  \regexNeg{\deriv\symbolb\regexr} = \deriv\symbolb{\regexNeg\regexr}$ by induction.

  If $\{\symbola, \symbolb\} \subseteq \literalbigcap\{\inv{\literall} \mid \literall\in \literalset_\regexr\}$, then $\{\symbola, \symbolb\} \in
  \inv{\bigcup\literalset_\regexr}$ so that, by induction, $\deriv\symbola\regexr \sqsubseteq
  \regexNull$ and $\deriv\symbolb\regexr \sqsubseteq
  \regexNull$. Hence, $\regexNeg{\deriv\symbola\regexr} =
  \regexNeg{\deriv\symbolb\regexr}$.
\end{proof}

\newpage
\section{Lemma~\ref{thm:first}: First and Next}
\label{sec:first-next}

For all $\regexr$,
$\bigcup \getFirst\regexr \supseteq \getNext\regexr$.

\begin{proof}[Proof of Lemma~\ref{thm:first}]
  The proof is by induction on $\regexr$.

  \textbf{Cases} $\regexEmpty$, $\seta$: trivial.

  \textbf{Case} $\regexOr\regexr\regexs$:
  Let $\literalset_\regexr = \getFirst\regexr$,
  $\literalset_\regexs = \getFirst\regexs$, and $\literalset = \literalset_\regexr \literaljoin
  \literalset_\regexs$. By induction, $\bigcup\literalset_\regexr\supseteq \getNext\regexr$ and
  $\bigcup\literalset_\regexs\supseteq \getNext\regexs$. By Lemma~\ref{lemma:properties-of-join},
  $\bigcup \literalset = \bigcup\literalset_\regexr \cup \bigcup\literalset_\regexs \supseteq
  \getNext\regexr \cup \getNext\regexs = \getNext{\regexOr\regexr\regexs}$.

  \textbf{Case} $\regexConcat\regexr\regexs$:
  Let $\literalset_\regexr = \getFirst\regexr$,
  $\literalset_\regexs = \getFirst\regexs$, and $\literalset = \literalset_\regexr \literaljoin
  \literalset_\regexs$.

  If $\neg\isNullable\regexr$, then $\bigcup \getFirst{\regexConcat\regexr\regexs} = \bigcup \getFirst\regexr
  \supseteq \getNext\regexr = \getNext{\regexConcat\regexr\regexs}$.

  If $\isNullable\regexr$, then  $\bigcup \getFirst{\regexConcat\regexr\regexs} = \bigcup
  (\getFirst\regexr \literaljoin \getFirst\regexs)
  \supseteq (\getNext\regexr \cup \getNext\regexs) = \getNext{\regexConcat\regexr\regexs}$ by
  induction and using Lemma~\ref{lemma:properties-of-join}.

  \textbf{Case} $\regexStar\regexr$:
  $\bigcup \getFirst{\regexStar\regexr} = \bigcup \getFirst\regexr \supseteq \getNext\regexr =
  \getNext{\regexStar\regexr}$ by induction

  \textbf{Case} $\regexAnd\regexr\regexs$:
  $\bigcup \getFirst{\regexAnd\regexr\regexs} = \bigcup (\getFirst\regexr \literalcap
  \getFirst\regexs)  =  \bigcup (\getFirst\regexr) \cap \bigcup (\getFirst\regexs) \supseteq
  \getNext\regexr \cap \getNext\regexs \supseteq
  \getNext{\regexAnd\regexr\regexs}$.

  \textbf{Case} $\regexNeg\regexr$:
  $\bigcup \getFirst{\regexNeg\regexr} = \regexAlphabet \supseteq \getNext{\regexNeg\regexr}$.
\end{proof}

\newpage
\section{Theorem~\ref{thm:firstderivative}: Left Quotient}
\label{sec:proof-firstderivative}

\begin{definition}[Next2]
  Let $\getFirstPlus{\regexr}~=~\getFirst{\regexr}\setminus\{\emptyset\}$ be the set of first literals of \ERE{} $\regexr$ exlcuding the ehe empty set $\{\emptyset\}$.
\end{definition}

\begin{proof}[Proof of Theorem~\ref{thm:firstderivative}]
  For any \ERE{} $\regexr$, for any literal $\literall\in\getFirstPlus{\regexr}$, and for any symbol $\symbola\in{\literall}$, the following equation holds:
  \begin{gather}
    \lang{\nderiv{\literall}{\regexr}} ~=~ \lang{\deriv{\symbola}{\regexr}}
  \end{gather}
  \begin{gather}
    \lang{\pderiv{\literall}{\regexr}} ~=~ \lang{\deriv{\symbola}{\regexr}}
  \end{gather}

  Proof by induction on $\regexr$.

  \begin{description}

    \item[Case] $\regexr=\regexEmpty$:~ Claim holds because $\lang{\nderiv{\literall}{\regexEmpty}} = \lang{\pderiv{\literall}{\regexEmpty}} = \lang{\deriv{\symbola}{\regexNull}} = \regexNull$.

    \item[Case] $\regexr=\setb$:~Claim holds because 
      \begin{align}
        \lang{\nderiv{\seta}{\setb}} ~=~ \lang{\deriv{\symbola}{\setb}} ~&=~ \begin{cases}
          \{\regexEmpty\}, & \seta\subseteq\setb\\
          \emptyset, & otherwise
        \end{cases}
      \end{align}
      and 
      \begin{align}
        \lang{\pderiv{\setb}{\seta}} ~=~ \lang{\deriv{\symbola}{\seta}} ~&=~ \begin{cases}
          \{\regexEmpty\}, & \seta\subseteq\setb\\
          \emptyset, & otherwise
        \end{cases}
      \end{align}

    \item[Case] $\regexr=\regexStar{\regexs}$:~ By induction 
      \begin{align}
        &\lang{\nderiv{\literall}{\regexs}} ~\eqIH~ \lang{\deriv{\symbola}{\regexs}}\\
        &\lang{\pderiv{\literall}{\regexs}} ~\eqIH~ \lang{\deriv{\symbola}{\regexs}}
      \end{align}
      holds.
      We obtain that
      \begin{align}
        &\forall\symbola:~ \lang{\deriv{\symbola}{\regexStar{\regexs}}} ~=~ \lang{\deriv{\symbola}{\regexs}\concat\regexStar{\regexs}}\\
        &\forall\literall:~ \lang{\nderiv{\literall}{\regexStar{\regexs}}} ~=~ \lang{\nderiv{\literall}{\regexs}\concat\regexStar{\regexs}}\\
        &\forall\literall:~ \lang{\pderiv{\literall}{\regexStar{\regexs}}} ~=~ \lang{\pderiv{\literall}{\regexs}\concat\regexStar{\regexs}}
      \end{align}
      holds. Claim holds because
      \begin{align}
        \forall\literall:~ \lang{\nderiv{\literall}{\regexStar{\regexs}}}
        ~=&~ \lang{\nderiv{\literall}{\regexs}\concat\regexStar{\regexs}}\\
        ~\eqIH&~ \{ \wordu\wordv ~|~ \wordu\in \lang{\deriv{\symbola}{\regexs}}, \wordv\in\lang{\regexStar{\regexs}}\} \\
        ~=&~ \lang{\deriv{\symbola}{\regexStar{\regexs}}}
      \end{align}
      and 
      \begin{align}
        \forall\literall:~ \lang{\pderiv{\literall}{\regexStar{\regexs}}}
        ~=&~ \lang{\pderiv{\literall}{\regexs}\concat\regexStar{\regexs}}\\
        ~\eqIH&~ \{ \wordu\wordv ~|~ \wordu\in \lang{\deriv{\symbola}{\regexs}}, \wordv\in\lang{\regexStar{\regexs}}\} \\
        ~=&~ \lang{\deriv{\symbola}{\regexStar{\regexs}}}
      \end{align}

    \item[Case] $\regexr=\regexOr{\regexs}{\regext}$:~ By induction 
      \begin{align}
        &\lang{\nderiv{\literall}{\regexs}} ~\eqIH~ \lang{\deriv{\symbola}{\regexs}}\\
        &\lang{\nderiv{\literall}{\regext}} ~\eqIH~ \lang{\deriv{\symbola}{\regext}}\\
        &\lang{\pderiv{\literall}{\regexs}} ~\eqIH~ \lang{\deriv{\symbola}{\regexs}}\\
        &\lang{\pderiv{\literall}{\regext}} ~\eqIH~ \lang{\deriv{\symbola}{\regext}}
      \end{align}
      holds.
      We obtain that
      \begin{align}
        &\lang{\deriv{\symbola}{\regexOr{\regexs}{\regext}}} ~=~ \lang{\deriv{\symbola}{\regexs}}\cup\lang{\deriv{\symbola}{\regext}}\\
        &\lang{\nderiv{\literall}{\regexOr{\regexs}{\regext}}} ~=~ \lang{\nderiv{\literall}{\regexs}}\cup\lang{\nderiv{\literall}{\regext}}\\
        &\lang{\pderiv{\literall}{\regexOr{\regexs}{\regext}}} ~=~ \lang{\pderiv{\literall}{\regexs}}\cup\lang{\pderiv{\literall}{\regext}}
      \end{align}
      holds. Claim holds because
      \begin{align}
        \lang{\nderiv{\literall}{\regexOr{\regexs}{\regext}}}
        ~=&~ \lang{\nderiv{\literall}{\regexs}}\cup\lang{\nderiv{\literall}{\regext}}\\
        ~\eqIH&~ \lang{\deriv{\symbola}{\regexs}} \cup \lang{\deriv{\symbola}{\regext}}\\
        ~=&~ \lang{\deriv{\symbola}{\regexOr{\regexs}{\regext}}}                        
      \end{align}
      and
      \begin{align}
        \lang{\pderiv{\literall}{\regexOr{\regexs}{\regext}}}
        ~=&~ \lang{\pderiv{\literall}{\regexs}}\cup\lang{\pderiv{\literall}{\regext}}\\
        ~\eqIH&~ \lang{\deriv{\symbola}{\regexs}} \cup \lang{\deriv{\symbola}{\regext}}\\
        ~=&~  \lang{\deriv{\symbola}{\regexOr{\regexs}{\regext}}}                       
      \end{align}

    \item[Case] $\regexr=\regexAnd{\regexs}{\regext}$:~ By induction
      \begin{align}
        &\lang{\nderiv{\literall}{\regexs}} ~\eqIH~ \lang{\deriv{\symbola}{\regexs}}\\
        &\lang{\nderiv{\literall}{\regext}} ~\eqIH~ \lang{\deriv{\symbola}{\regext}}\\
        &\lang{\pderiv{\literall}{\regexs}} ~\eqIH~ \lang{\deriv{\symbola}{\regexs}}\\
        &\lang{\pderiv{\literall}{\regext}} ~\eqIH~ \lang{\deriv{\symbola}{\regext}}
      \end{align}
      holds.
      We obtain that
      \begin{align}
        &\lang{\deriv{\symbola}{\regexAnd{\regexs}{\regext}}} ~=~ \lang{\deriv{\symbola}{\regexs}}\cap\lang{\deriv{\symbola}{\regext}}\\
        &\lang{\nderiv{\literall}{\regexAnd{\regexs}{\regext}}} ~=~ \lang{\nderiv{\literall}{\regexs}}\cap\lang{\nderiv{\literall}{\regext}}\\
        &\lang{\pderiv{\literall}{\regexAnd{\regexs}{\regext}}} ~=~ \lang{\pderiv{\literall}{\regexs}}\cap\lang{\pderiv{\literall}{\regext}}
      \end{align}
      holds. Claim holds because
      \begin{align}
        \lang{\nderiv{\literall}{\regexAnd{\regexs}{\regext}}}
        ~=&~ \lang{\nderiv{\literall}{\regexs}}\cap\lang{\nderiv{\literall}{\regext}}\\
        ~\eqIH&~ \lang{\deriv{\symbola}{\regexs}} \cap \lang{\deriv{\symbola}{\regext}}\\
        ~=&~ \lang{\deriv{\symbola}{\regexAnd{\regexs}{\regext}}}                       
      \end{align}
      and
      \begin{align}
        \lang{\pderiv{\literall}{\regexAnd{\regexs}{\regext}}}
        ~=&~ \lang{\pderiv{\literall}{\regexs}}\cap\lang{\pderiv{\literall}{\regext}}\\
        ~\eqIH&~ \lang{\deriv{\symbola}{\regexs}} \cap \lang{\deriv{\symbola}{\regext}}\\
        ~=&~ \lang{\deriv{\symbola}{\regexAnd{\regexs}{\regext}}}                       
      \end{align}

    \item[Case] $\regexr=\regexNeg{\regexs}$:~ By induction 
      \begin{align}
        &\lang{\nderiv{\literall}{\regexs}} ~\eqIH~ \lang{\deriv{\symbola}{\regexs}}\\
        &\lang{\pderiv{\literall}{\regexs}} ~\eqIH~  \lang{\deriv{\symbola}{\regexs}}
      \end{align}
      holds.
      We obtain that
      \begin{align}
        &\forall\symbola:~ \lang{\deriv{\symbola}{\regexNeg{\regexs}}} ~=~ \regexWords\backslash\lang{\deriv{\symbola}{\regexs}}\\
        &\forall\literall:~ \lang{\nderiv{\literall}{\regexNeg{\regexs}}} ~=~ \regexWords\backslash\lang{\pderiv{\literall}{\regexs}}\\
        &\forall\literall:~ \lang{\pderiv{\literall}{\regexNeg{\regexs}}} ~=~ \regexWords\backslash\lang{\nderiv{\literall}{\regexs}}
      \end{align}
      holds. Claim holds because
      \begin{align}
        \forall\literall:~ \lang{\nderiv{\literall}{\regexNeg{\regexs}}}
        ~=&~ \regexWords\backslash\lang{\pderiv{\literall}{\regexs}}\\
        ~\eqIH&~ \regexWords\backslash \lang{\deriv{\symbola}{\regexs}}\\
        ~=&~ \lang{\deriv{\symbola}{\regexNeg{\regexs}}}
      \end{align}
      and 
      \begin{align}
        \forall\literall:~ \lang{\pderiv{\literall}{\regexNeg{\regexs}}}
        ~=&~ \regexWords\backslash\lang{\nderiv{\literall}{\regexs}}\\
        ~\eqIH&~ \regexWords\backslash \lang{\deriv{\symbola}{\regexs}}\\
        ~=&~ \lang{\deriv{\symbola}{\regexNeg{\regexs}}}
      \end{align}

    \item[Case] $\regexr=\regexConcat{\regexs}{\regext}$:~ By induction
      \begin{align}
        &\lang{\nderiv{\literall}{\regexs}} ~\eqIH~ \lang{\deriv{\symbola}{\regexs}}\\
        &\lang{\nderiv{\literall}{\regext}} ~\eqIH~ \lang{\deriv{\symbola}{\regext}}\\
        &\lang{\pderiv{\literall}{\regexs}} ~\eqIH~ \lang{\deriv{\symbola}{\regexs}}\\
        &\lang{\pderiv{\literall}{\regext}} ~\eqIH~ \lang{\deriv{\symbola}{\regext}}
      \end{align}
      holds.
      We obtain that 
      \begin{align}
        &\forall\symbola:~ \lang{\deriv{\symbola}{\regexConcat{\regexs}{\regext}}} ~=~ \begin{cases}
          \lang{\deriv{\symbola}{\regexs}\concat\regext} \cup \lang{\deriv{\symbola}{\regext}}, & \isNullable{\regexs}\\
          \lang{\deriv{\symbola}{\regexs}\concat\regext}, & otherwise
        \end{cases}\\
        &\forall\literall:~ \lang{\nderiv{\literall}{\regexConcat{\regexs}{\regext}}} ~=~ \begin{cases}
          \lang{\nderiv{\literall}{\regexs}\concat\regext} \cup \lang{\nderiv{\literall}{\regext}}, & \isNullable{\regexs}\\
          \lang{\nderiv{\literall}{\regexs}\concat\regext}, & otherwise
        \end{cases}\\
        &\forall\literall:~ \lang{\pderiv{\literall}{\regexConcat{\regexs}{\regext}}} ~=~ \begin{cases}
          \lang{\pderiv{\literall}{\regexs}\concat\regext} \cup \lang{\pderiv\literall{\symbola}{\regext}}, & \isNullable{\regexs}\\
          \lang{\pderiv{\literall}{\regexs}\concat\regext}, & otherwise
        \end{cases}
      \end{align}
      holds.
      \begin{description}
        \item[Subcase] $\isNullable{\regexs}$:~
          Claim holds because
          \begin{align}
            \lang{\nderiv{\literall}{\regexConcat{\regexs}{\regext}}}
            ~=&~ \lang{\nderiv{\literall}{\regexs}\concat\regext} \cup \lang{\nderiv{\literall}{\regext}}\\
            ~\eqIH&~ \{ \wordu\wordv ~|~ \wordu\in \lang{\deriv{\symbola}{\regexs}}, \wordv\in\lang{\regext}\} \cup  \lang{\deriv{\symbola}{\regext}}\\
            ~=&~ \lang{\deriv{\symbola}{\regexs}\concat\regext} \cup \lang{\deriv{\symbola}{\regext}}\\
            ~=&~ \deriv{\symbola}{\regexs\concat\regext}
          \end{align}
          and
          \begin{align}
            \lang{\pderiv{\literall}{\regexConcat{\regexs}{\regext}}}
            ~=&~ \lang{\pderiv{\literall}{\regexs}\concat\regext} \cup \lang{\pderiv{\literall}{\regext}}\\
            ~\eqIH&~ \{ \wordu\wordv ~|~ \wordu\in \lang{\deriv{\symbola}{\regexs}}, \wordv\in\lang{\regext}\} \cup \lang{\deriv{\symbola}{\regext}}\\
            ~=&~ \lang{\deriv{\symbola}{\regexs}\concat\regext} \cup \lang{\deriv{\symbola}{\regext}}\\
            ~=&~ \deriv{\symbola}{\regexs\concat\regext}
          \end{align}
        \item[Subcase] $\neg\isNullable{\regexs}$:~
          Claim holds because
          \begin{align}
            \lang{\nderiv{\literall}{\regexConcat{\regexs}{\regext}}}
            ~=&~ \lang{\nderiv{\literall}{\regexs}\concat\regext}\\
            ~\eqIH&~ \{ \wordu\wordv ~|~ \wordu\in \lang{\deriv{\symbola}{\regexs}}, \wordv\in\lang{\regext}\}\\
            ~=&~ \lang{\deriv{\symbola}{\regexs}\concat\regext}\\
            ~=&~ \deriv{\symbola}{\regexs\concat\regext}
          \end{align}
          and
          \begin{align}
            \lang{\pderiv{\literall}{\regexConcat{\regexs}{\regext}}}
            ~=&~ \lang{\pderiv{\literall}{\regexs}\concat\regext}\\
            ~\eqIH&~ \{ \wordu\wordv ~|~ \wordu\in \lang{\deriv{\symbola}{\regexs}}, \wordv\in\lang{\regext}\}\\
            ~=&~ \lang{\deriv{\symbola}{\regexs}\concat\regext}\\
            ~=&~ \deriv{\symbola}{\regexs\concat\regext}
          \end{align}

      \end{description}
  \end{description}
\end{proof}

\newpage
\section{Theorem~\ref{thm:symbol-containment}: Semantic Containment}
\label{sec:proof-semantic-containment}

\begin{lemma}[Word Inclusion]\label{thm:word-inclusion}
  For all \ERE{} $\regexr$ and words $\wordw$ in $\regexWords$,
  \begin{displaymath}
    \wordw\in\lang{\regexr} \Leftrightarrow \isNullable{\deriv{\wordw}{\regexr}}
  \end{displaymath}
\end{lemma}

\begin{proof}[Proof of Lemma~\ref{thm:word-inclusion}]
  Proof by the definition of $\delta$ and $\nullable$.
\end{proof}

\begin{lemma}[Word Containment]\label{thm:word-containment}
  For all \ERE{} $\regexr$ and $\regexs$,
  \begin{displaymath}
    \isSubSetOf{\regexr}{\regexs} \Leftrightarrow \isNullable{\deriv{\wordw}{\regexs}} ~\text{for all}~ \wordw\in\lang{\regexr}
  \end{displaymath}
\end{lemma}

\begin{proof}[Proof of Lemma~\ref{thm:word-containment}]
  An \ERE{} $\regexr$ is subset of another \ERE{} $\regexs$ iff for all words $\wordw\in\lang{\regexr}$ the derivation of $\regexs$ w.r.t. word $\wordw$ is nullable. For all $\wordw\in\regexWords$ it holds that $\wordw\in\lang{\regexs}$ iff $\isNullable{\deriv{\wordw}{\regexs}}$. It is trivial to see that
  \begin{align}
    &~ \isSubSetOf{\regexr}{\regexs}\\
    ~\Leftrightarrow&~ \lang{\regexr}\subseteq\lang{\regexs}\\
    ~\Leftrightarrow&~ \forall \wordw\in\lang{\regexr}: \wordw\in\lang{\regexs}\\
    ~\Leftrightarrow&~ \forall \wordw\in\lang{\regexr}: \isNullable{\deriv{\wordw}{\regexs}}
  \end{align}
  holds.
\end{proof}

\begin{proof}[Proof of Theorem~\ref{thm:symbol-containment}]
  For all regular expressions $\regexr$ and $\regexs$,
  \begin{displaymath}
    \isSubSetOf{\regexr}{\regexs} \Leftrightarrow (\isNullable{\regexr}\Rightarrow\isNullable{\regexs}) ~\wedge~ (\forall\symbola\in\getNext{\regexr})~\isSubSetOf{\deriv\symbola\regexr}{\deriv\symbola\regexs} 
  \end{displaymath}
  An \ERE{} $\regexr$ is subset of another \ERE{} $\regexs$ iff for all symbols $\symbola$ in $\getNext{\regexr}$ the derivation of $\regexr$ w.r.t. symbol $\symbola$ is subset of the derivation of $\regexs$ w.r.t. $\symbola$. 
  We obtain that
  \begin{align}
    \lang{\deriv{\symbola}{\regexr}} = \leftquotient{\symbola}{\lang{\regexr}}	
  \end{align}
  and this leads to 
  \begin{align}
    \{\regexEmpty ~|~ \isNullable{\regexr} \} \cup
    \{ \symbola\wordu ~|~ \symbola\in\lang{\getNext{\regexr}}, \wordu\in\leftquotient{\symbola}{\lang{\regexr}} \} &= \lang{\regexr}
  \end{align}
  Claim holds because
  \begin{align}
    &~ \isSubSetOf{\regexr}{\regexs}\\
    ~\Leftrightarrow&~ \lang{\regexr}\subseteq\lang{\regexs}\\
    ~\Leftrightarrow&~ \forall \wordu\in\lang{\regexr}:~ \wordu\in\lang{\regexs}\\
    ~\Leftrightarrow&~ \regexEmpty\in\lang{\regexr} \Rightarrow \regexEmpty\in\lang{\regexs} ~\wedge~ \forall \symbola, \wordu:~ \symbola\wordu\in\lang{\regexr} \Rightarrow \isNullable{\deriv{\symbola\wordu}{\regexs}}\\
    ~\Leftrightarrow&~ \isNullable{\regexr} \Rightarrow \isNullable{\regexs} ~\wedge~ \forall \symbola\in\getNext{\regexr}, \forall\wordu:~ \symbola\wordu\in\lang{\regexr} \Rightarrow \isNullable{\deriv{\wordu}{\deriv{\symbola}{\regexs}}}\\
    ~\Leftrightarrow&~ \isNullable{\regexr} \Rightarrow \isNullable{\regexs} ~\wedge~ \forall \symbola\in\getNext{\regexr}, \forall\wordu \in\lang{\deriv{\symbola}{\regexr}}:~ \isNullable{\deriv{\wordu}{\deriv{\symbola}{\regexs}}}\\
    ~\Leftrightarrow&~ \isNullable{\regexr} \Rightarrow \isNullable{\regexs} ~\wedge~ \forall \symbola\in\getNext{\regexr}:~ \lang{\deriv{\symbola}{\regexr}}\subseteq\lang{\deriv{\symbola}{\regexs}}\\
    ~\Leftrightarrow&~ \isNullable{\regexr} \Rightarrow \isNullable{\regexs} ~\wedge~ \forall \symbola\in\getNext{\regexr}:~ \isSubSetOf{\deriv{\symbola}{\regexr}}{\deriv{\symbola}{\regexs}}
  \end{align}
\end{proof}

\newpage
\section{Theorem~\ref{thm:containment}: Symbolic Containment}
\label{sec:proof-symbolic-containment}

\begin{proof}[Proof of Theorem~\ref{thm:containment}]
  The proof is by contraposition.
  If $\isNotSubSetOf{\regexr}{\regexs}$ then $\exists\literall\in\getFirst{\isSubSetOf{\regexr}{\regexs}}:~ \isNotSubSetOf{\nderiv{\literall}{\regexr}}{\nderiv{\literall}{\regexs}}$ or $\neg(\isNullable{\regexr}\Rightarrow\isNullable{\regexs})$.

  We obtain that:
  \begin{align}
    \isNotSubSetOf{\regexr}{\regexs} ~&\Leftrightarrow~ \lang{\regexr}\nsubseteq\lang{\regexs}\\
    ~&\Leftrightarrow~ \exists\wordu\in\lang{\regexr}\backslash\lang{\regexs}
  \end{align}

  \begin{description}
    \item[Case] $\wordu=\regexEmpty$:~ \\
      Claim holds because $\neg(\isNullable{\regexr}\Rightarrow\isNullable{\regexs})$.
    \item[Case] $\wordu\neq\regexEmpty$:~\\
      It must be that $\wordu=\symbola\wordv$ with $\symbola\in\getNext{\regexr}={\getFirst{\regexr}}$.
      Therefore $\exists\literall\in\getFirst{\regexr}:~ \symbola\in{\literall}$.
      \begin{description}
        \item[Subcase] $\symbola\notin\getNext{\regexs}$:~\\
          Claim holds by Lemma~\ref{thm:firstderivative} and Lemma~\ref{thm:preservation} because $\exists\literall\in\getFirst{\regexr}:$ $\nderiv{\literall}{\regexr}\neq\regexNull$ and $\nderiv{\literall}{\regexs}=\regexNull$ implies that $\isNotSubSetOf{\nderiv{\literall}{\regexr}}{\nderiv{\literall}{\regexs}}$.

        \item[Subcase] $\symbola\in\getNext{\regexs}$:~\\ 
          By Lemma~\ref{thm:firstderivative} and Lemma~\ref{thm:preservation} claim holds because
          \mbox{$\wordv\in\lang{\deriv{\symbola}{\regexr}}\backslash\lang{\deriv{\symbola}{\regexs}}$} implies that
          $\wordv\in\lang{\nderiv{\literall}{\regexr}}\backslash\lang{\nderiv{\literall}{\regexs}}$

      \end{description}
  \end{description}
\end{proof}

\newpage
\section{Theorem~\ref{thm:soundness}: Soundness}
\label{sec:proof-soundness}

For all regular expression $\regexr$ and $\regexs$:
\begin{align*}
  \emptyset~\entails~\checkSubSetOf{\regexr}{\regexs}~:~\top~\Leftrightarrow~\isSubSetOf{\regexr}{\regexs}
\end{align*}

\begin{proof}
  We prove that $\ccCtx \entails \checkSubSetOf{\regexr}{\regexs} ~:~ \bot$ iff $\isNotSubSetOf{\regexr}{\regexs}$, for
  all contexts $\ccCtx$ where $\checkSubSetOf{\regexr}{\regexs} \notin \ccCtx$. 
  This is sufficient because each regular
  inequality gives rise to a finite derivation by
  Theorem~\ref{thm:finiteness}.

  The ``only-if'' direction is by rule induction on the derivation of $\ccCtx \entails
  \checkSubSetOf{\regexr}{\regexs} ~:~ \perp$. 
  \begin{itemize}
    \item Suppose the last rule is \Rule{\RuleCCDisprove}. By inversion, $\isNullable\regexr$ and
      $\neg\isNullable\regexs$ so that $\isNotSubSetOf{\regexr}{\regexs}$.
    \item Suppose the last rule is \Rule{\RuleCCUnfoldFalse}. By inversion, 
      \begin{gather}
        \notInCcCtx{\checkSubSetOf{\regexr}{\regexs}}\\
        \isNullable\regexr \Rightarrow \isNullable\regexs\\
        \exists\literall\in\getFirst{\checkSubSetOf{\regexr}{\regexs}}:
        ~\ccCtx \cup\{\checkSubSetOf{\regexr}{\regexs}\}~\entails~
        \checkSubSetOf{\deriv{\literall}{\regexr}}{\deriv{\literall}{\regexs}}~:~\perp%
      \end{gather}
      By induction, we obtain that $\isNotSubSetOf{\deriv{\literall}{\regexr}}{\deriv{\literall}{\regexs}}$, for some
      $\literall\in\getFirst{\checkSubSetOf{\regexr}{\regexs}}$. By Theorem~\ref{thm:containment}, we obtain
      $\isNotSubSetOf\regexr\regexs$. 
  \end{itemize}

  For the ``if'' direction, the assumption that $\isNotSubSetOf\regexr\regexs$ implies that $\lang{\regexr} \setminus
  \lang{\regexs} \ne \emptyset$. Let $\wordu \in \lang{\regexr} \setminus \lang{\regexs}$ a word of shortest length.
  We continue by induction on $\wordu$. If
  $\wordu=\regexEmpty$, then $\isNullable{\regexr}$ but not
  $\isNullable{\regexs}$ must hold. By our assumption on $\ccCtx$, it cannot be that $\checkSubSetOf\regexr\regexs \notin
  \ccCtx$. By rule $\Rule{\RuleCCDisprove}$,  $\ccCtx \entails \checkSubSetOf\regexr\regexs ~:~\bot$.

  If $\wordu = \symbola\wordu'$, then there exists $\literall_\symbola \in
  \getFirst{\checkSubSetOf{\regexr}{\regexs}}$ such that $\symbola\in{\literall_\symbola}$ (by Lemma~\ref{thm:first}).
  It must be that $\isNullable\regexr \Rightarrow \isNullable\regexs$: otherwise, we get a contradiction
  against the minimality of $\wordu$'s length.
  By Theorem~\ref{thm:containment} it must be that
  $\isSubSetOf{\deriv{\literall_\symbola}{\regexr} \not}{\deriv{\literall_\symbola}{\regexs}}$ so that induction yields 
  $\ccCtx
  \cup\{\checkSubSetOf{\regexr}{\regexs}\}\\~\entails~\checkSubSetOf{\deriv{\literall_\symbola}{\regexr}}{\deriv{\literall_\symbola}{\regexs}}~:~\bot$.
  Applying rule $\Rule{\RuleCCUnfoldFalse}$ yields
  $\ccCtx\entails\checkSubSetOf{\regexr}{\regexs}~:~\bot$.
\end{proof}

\newpage
\section{Lemma~\ref{thm:preservation}: Coverage}

\begin{lemma}[Coverage]\label{thm:preservation}
  For all symbols $\symbola\in\regexAlphabet$, words $\wordu\in\regexWords$, and \ERE{}s on $\regexAlphabet$ it holds that:
  \begin{align*}
    \wordu\in\lang{\deriv{\symbola}{\regexr}}~\Leftrightarrow~\exists\literall\in\getFirstPlus{\regexr}:~\wordu\in\lang{\pderiv{\literall}{\regexr}}\\
    \wordu\in\lang{\deriv{\symbola}{\regexr}}~\Leftrightarrow~\exists\literall\in\getFirstPlus{\regexr}:~\wordu\in\lang{\nderiv{\literall}{\regexr}}
  \end{align*}
\end{lemma}

\begin{proof}[Proof of Lemma~\ref{thm:preservation}]
  Suppose $\lang{\deriv{\symbola}{\regexr}}\neq\regexNull$.
  Because $\nderiv{\literall}{\regexr}$=$\pderiv{\literall}{\regexr}$ for all $\literall\in\getFirstPlus{\regexr}$ show $\exists\literall\in\getFirstPlus{\regexr}:~$ $\wordw\in\lang{\nderiv{\literall}{\regexr}}$.
  Proof by induction on $\regexr$.

  \begin{description}

    \item[Case] $\regexr=\regexEmpty$, $\getFirstPlus{\regexr}=\emptyset$: Contradicts assumption.

    \item[Case] $\regexr=\seta$, $\getFirstPlus{\regexr}=\{\seta\}$:~\\
      We obtain that $\symbola\in{\seta} ~\Rightarrow~ \deriv{\symbola}{\seta}=\regexEmpty$. Claim holds because $\getFirstPlus{\regexr}=\{\seta\}$, $\nderiv{\seta}{\seta}=\regexEmpty$, and thus $\wordw=\regexEmpty$ and $\regexEmpty\in\lang{\regexEmpty}$.

    \item[Case] $\regexr=\regexStar{\regexs}$, $\getFirstPlus{\regexr}=\getFirstPlus{\regexs}$:~\\
      We obtain that $\wordw\in\lang{\deriv{\symbola}{\regexStar{\regexs}}}=\lang{\deriv{\symbola}{\regexs}\regexConcatOp\regexStar{\regexs}}\neq\regexNull$. By induction $\exists\literall'\in\getFirstPlus{\regexs}:~ \wordu\in\lang{\nderiv{\literall'}{\regexs}}$.
      The chain holds because $\getFirstPlus{\regexStar{\regexs}}=\getFirstPlus{\regexs}$ and $\nderiv{\literall}{\regexStar{\regexs}}=\nderiv{\literall}{\regexs}\regexConcatOp\regexs\regexStarOp$ and $\wordu\in\lang{\nderiv{\literall}{\regexs}}, \wordv\in\lang{\nderiv{\literall}{\regexStar{\regexs}}}$ implies $\wordw=\wordu\regexConcatOp\wordv\in\lang{\nderiv{\literall}{\regexStar{\regexs}}}$.

    \item[Case] $\regexr=(\regexs\regexOrOp\regext)$, $\getFirstPlus{\regexr}=\getFirstPlus{\regexs} \literaljoin \getFirstPlus{\regext}$:~\\
      We obtain that $\wordw\in\lang{\deriv{\symbola}{\regexs\regexOrOp\regext}}=\lang{\deriv{\symbola}{\regexs}}\cup\lang{\deriv{\symbola}{\regext}}\neq\regexNull$. 
      By induction $\exists\literall'\in\getFirstPlus{\regexs}:~ \wordu\in\lang{\nderiv{\literall'}{\regexs}}$ and 
      $\exists\literall''\in\getFirstPlus{\regext}:~ \wordv\in\lang{\nderiv{\literall''}{\regext}}$.
      The chain holds because $\getFirstPlus{\regexs\regexOrOp\regext}=\getFirstPlus{\regexs} \literaljoin \getFirstPlus{\regext}$ and $\nderiv{\literall}{\regexs\regexOrOp\regext}=\nderiv{\literall}{\regexs}\regexOrOp\nderiv{\literall}{\regext}$ and $\wordw\in\lang{\nderiv{\literall}{\regexs}}$ or $\wordw\in\lang{\nderiv{\literall}{\regext}}$ implies $\wordw\in\lang{\nderiv{\literall}{\regexs\regexOrOp\regext}}$.

    \item[Case] $\regexr=(\regexs\regexAndOp\regext)$, $\getFirstPlus{\regexr}=\{ \literall'\literalcap\literall'' ~|~ \literall'\in\getFirstPlus{\regexs}, \literall''\in\getFirstPlus{\regext} \} $:~ \\
      We obtain that $\wordw\in\lang{\deriv{\symbola}{\regexs\regexAndOp\regext}}=\lang{\deriv{\symbola}{\regexs}}\cap\lang{\deriv{\symbola}{\regext}}$ implies $\wordw\in\lang{\deriv{\symbola}{\regexs}}$ and $\wordw\in\lang{\deriv{\symbola}{\regext}}$. 
      By induction $\exists\literall'\in\getFirstPlus{\regexs}:~ \wordw\in\lang{\nderiv{\literall}{\regexs}}$ and 
      $\exists\literall''\in\getFirstPlus{\regext}:~ \wordw\in\lang{\nderiv{\literall''}{\regext}}$.
      Let $\literall=\literall'\literalcap\literall''\in\getFirstPlus{\regexs\regexAndOp\regext}$. If $\symbola\in{\literall'}$ and $\symbola\in{\literall'}$ then $\symbola\in{\literall}$.
      The chain holds because $\getFirstPlus{\regexs\regexAndOp\regext}=\{\literall'\literalcap\literall'' ~|~ \literall'\in\getFirstPlus{\regexs}, \literall''\in\getFirstPlus{\regext} \}$ and $\nderiv{\literall}{\regexs\regexAndOp\regext}=\nderiv{\literall}{\regexs}\regexAndOp\nderiv{\literall}{\regext}$, and $\wordw\in\lang{\nderiv{\literall}{\regexs}}$ and $\wordw\in\lang{\nderiv{\literall}{\regext}}$ implies $\wordw\in\lang{\nderiv{\literall}{\regexs\regexAndOp\regext}}$.

    \item[Case] $\regexr=(\regexNeg\regexs)$, $\getFirstPlus{\regexr}=\literalbigcap\{\inv{\literall}~|~\literall\in\getFirstPlus{\regexs}\}~\literaljoin~\{\literall\in\getFirstPlus{\regexs}~|~\nderiv{\literall}{\regexs}\neq\regexWords\}$:~ \\
      We obtain that $\wordw\in\lang{\deriv{\symbola}{\regexNeg\regexs}}=\regexWords\setminus\lang{\deriv{\symbola}{\regexs}}$ implies $\wordw\not\in\lang{\deriv{\symbola}{\regexs}}$. By induction $\exists\literall'\in\getFirstPlus{\regexs}:~ \wordw\in\lang{\nderiv{\literall}{\regexs}}$.
      Let $\literall=\literalbigcap\{\inv{\literall}~|~\literall\in\getFirstPlus{\regexs}\}~\literaljoin~\{\literall\in\getFirstPlus{\regexs}~|~\nderiv{\literall}{\regexs}\neq\regexWords\}$. If $\lang{\deriv{\symbola}{\regexNeg\regexs}}\neq\emptyset$ implies $\lang{\deriv{\symbola}{\regexs}}\neq\regexWords$.
      The chain holds because $\getFirstPlus{\regexNeg\regexs}=\literalbigcap\{\inv{\literall}~|~\literall\in\getFirstPlus{\regexs}\}~\literaljoin~\{\literall\in\getFirstPlus{\regexs}~|~\nderiv{\literall}{\regexs}\neq\regexWords\}$ and $\nderiv{\literall}{\regexNeg\regexs}=\regexNeg\nderiv{\literall}{\regexs}$, and $\wordw\not\in\lang{\nderiv{\literall}{\regexs}}$ implies $\wordw\in\lang{\nderiv{\literall}{\regexNeg\regexs}}$.

    \item[Case] $\regexr=(\regexs\regexConcatOp\regext)$:~
      \begin{description}
        \item[Subcase] $\isNullable{\regexs}$, $\getFirstPlus{\regexr}=\getFirstPlus{\regexs}\literaljoin\getFirstPlus{\regext}$:~\\
          We obtain that $\wordw\in\lang{\deriv{\symbola}{\regexs\regexConcatOp\regext}}=\lang{\deriv{\symbola}{\regexs}\regexConcatOp\regext}\cup\lang{\deriv{\symbola}{\regext}}$ implies $\wordw\in\lang{\deriv{\symbola}{\regexs}\regexConcatOp\regext}$ or $\wordw\in\lang{\deriv{\symbola}{\regext}}$.
          By induction $\exists\literall'\in\getFirstPlus{\regexs}:~ \wordu\in\lang{\nderiv{\literall}{\regexs}}$ and 
          $\exists\literall''\in\getFirstPlus{\regext}:~ \wordv\in\lang{\nderiv{\literall''}{\regext}}$.
          The chain holds because $\getFirstPlus{\regexs\regexConcatOp\regext}= \getFirstPlus{\regexs}\literaljoin\getFirstPlus{\regext}$ and $\nderiv{\literall}{\regexs\regexConcatOp\regext}=(\nderiv{\literall}{\regexs}\regexConcatOp\regext)\regexOrOp\nderiv{\literall}{\regext}$, and $\wordu\in\lang{\nderiv{\literall}{\regexs}}$ and $\wordv\in\lang{\regext}$ implies $\wordw=\wordu\concat\wordv\in\lang{\nderiv{\literall}{\regexs\regexConcatOp\regext}}$ or $\wordw=\regexEmpty\concat\wordv\in\lang{\nderiv{\literall}{\regexs\regexConcatOp\regext}}$.

        \item[Subcase] $\neg\isNullable{\regexs}$, $\getFirstPlus{\regexr}=\getFirstPlus{\regexs}$:~\\
          We obtain that $\wordw\in\lang{\deriv{\symbola}{\regexs\regexConcatOp\regext}}=\lang{\deriv{\symbola}{\regexs}\regexConcatOp\regext}$ implies $\wordw\in\lang{\deriv{\symbola}{\regexs}\regexConcatOp\regext}$.
          By induction $\exists\literall'\in\getFirstPlus{\regexs}:~ \wordu\in\lang{\nderiv{\literall}{\regexs}}$.
          The chain holds because $\getFirstPlus{\regexs\regexConcatOp\regext}= \getFirstPlus{\regexs}$ and $\nderiv{\literall}{\regexs\regexConcatOp\regext}=\nderiv{\literall}{\regexs}\regexConcatOp\regext$, and $\wordu\in\lang{\nderiv{\literall}{\regexs}}$ and $\wordv\in\lang{\nderiv{\literall}{\regext}}$ implies $\wordw=\wordu\concat\wordv\in\lang{\nderiv{\literall}{\regexs\regexConcatOp\regext}}$.
      \end{description}
  \end{description}
\end{proof}

\newpage
\section{Lemma~\ref{thm:equivalence}: Equivalence}

\begin{lemma}[Equivalence]\label{thm:equivalence}
  For all \ERE{} $\regexr$, literals $\literall\in\getFirst{\regexr}$, and literals $\literall'$ with $\literall'\subseteq\literall~|~\literall'\neq\emptyset$ holds:
  \begin{align*}
    \lang{\pderiv{\literall}{\regexr}}~\Leftrightarrow~\lang{\pderiv{\literall'}{\regexr}}\\
    \lang{\nderiv{\literall}{\regexr}}~\Leftrightarrow~\lang{\nderiv{\literall'}{\regexr}}\\
  \end{align*}
\end{lemma}

\begin{proof}[Proof of Lemma \ref{thm:preservation}]
  Suppose $\getFirst{\regexr}\neq\{\emptyset\}$.
  Because $\nderiv{\literall}{\regexr}$=$\pderiv{\literall}{\regexr}$ for all $\literall\in\getFirst{\regexr}$ show $\lang{\nderiv{\literall}{\regexr}}=\lang{\nderiv{\literall'}{\regexr}}$. Proof by induction on $\regexr$.

  \begin{description}

    \item[Case] $\regexr=\regexEmpty$, $\getFirst{\regexr}=\{\emptyset\}$: Contradicts assumption.

    \item[Case] $\regexr=\seta$, $\getFirst{\regexr}=\{\seta\}$:~\\
      Claim holds because for all $\literall'\subseteq\seta~\Rightarrow~\nderiv{\literall'}{\regexr}=\nderiv{\literall}{\regexr}=\regexEmpty$ and thus $\lang{\nderiv{\literall}{\regexr}}=\lang{\nderiv{\literall'}{\regexr}}$.  

    \item[Case] $\regexr=\regexStar{\regexs}$, $\getFirst{\regexr}=\getFirst{\regexs}$:~\\
      We obtain that $\lang{\nderiv{\literall}{\regexStar{\regexs}}}=\lang{\nderiv{\literall}{\regexs}\regexConcatOp\regexStar{\regexs}}\neq\regexNull$. By induction $\forall\literall_\regexs\in\getFirst{\regexs},\literall_\regexs'\subset\literall_\regexs:$ $\lang{\nderiv{\literall_\regexs}{\regexs}}=\lang{\nderiv{\literall_\regexs'}{\regexs}}$.
      The chain holds because $\getFirst{\regexStar{\regexs}}=\getFirst{\regexs}$ and $\lang{\nderiv{\literall'}{\regexStar{\regexs}}}=\lang{\nderiv{\literall'}{\regexs}\regexConcatOp\regexs\regexStarOp}$.

    \item[Case] $\regexr=(\regexs\regexOrOp\regext)$, $\getFirst{\regexr}=\getFirst{\regexs} \literaljoin \getFirst{\regext}$:~\\
      We obtain that $\lang{\nderiv{\literall}{\regexs\regexOrOp\regext}}=\lang{\nderiv{\literall}{\regexs}}\cup\lang{\nderiv{\literall}{\regext}}\neq\regexNull$. 
      By induction $\forall\literall_\regexs\in\getFirst{\regexs},\literall_\regexs'\subset\literall_\regexs:$ $\lang{\nderiv{\literall_\regexs}{\regexs}}=\lang{\nderiv{\literall_\regexs'}{\regexs}}$ and $\forall\literall_\regext\in\getFirst{\regext},\literall_\regext'\subset\literall_\regext:$ $\lang{\nderiv{\literall_\regext}{\regext}}=\lang{\nderiv{\literall_\regext'}{\regext}}$.
      The chain holds because $\getFirst{\regexs\regexOrOp\regext}=\getFirst{\regexs} \literaljoin \getFirst{\regext}$ and $\forall\literall''\in\getFirst{\regexr}\cup\getFirst{\regexs}:~\exists\literall'''\in\getFirst{\regexs\regexOrOp\regext}:~\literall'''\subseteq\literall''$ and $\lang{\nderiv{\literall'}{\regexs\regexOrOp\regext}}=\lang{\nderiv{\literall'}{\regexs}}\cup\lang{\nderiv{\literall'}{\regext}}$.

    \item[Case] $\regexr=(\regexs\regexAndOp\regext)$, $\getFirst{\regexr}=\{ \literall'\literalcap\literall'' ~|~ \literall'\in\getFirst{\regexs}, \literall''\in\getFirst{\regext} \} $:~ \\
      We obtain that $\lang{\nderiv{\literall}{\regexs\regexAndOp\regext}}=\lang{\nderiv{\literall}{\regexs}}\cap\lang{\nderiv{\literall}{\regext}}\neq\regexNull$.
      By induction $\forall\literall_\regexs\in\getFirst{\regexs},\literall_\regexs'\subset\literall_\regexs:$ $\lang{\nderiv{\literall_\regexs}{\regexs}}=\lang{\nderiv{\literall_\regexs'}{\regexs}}$ and $\forall\literall_\regext\in\getFirst{\regext},\literall_\regext'\subset\literall_\regext:$ $\lang{\nderiv{\literall_\regext}{\regext}}=\lang{\nderiv{\literall_\regext'}{\regext}}$.   
      Let $\literall=\literall_\regexs\literalcap\literall_\regext\in\getFirst{\regexs\regexAndOp\regext}$. If $\literall'\subseteq\literall$ then $\literall'\subseteq\literall_\regexs$ and $\literall'\subseteq\literall_\regext$.
      The chain holds because $\getFirst{\regexs\regexAndOp\regext}=\getFirst{\regexs} \sqcap \getFirst{\regext}$ and $\lang{\nderiv{\literall'}{\regexs\regexAndOp\regext}}=\lang{\nderiv{\literall'}{\regexs}}\cap\lang{\nderiv{\literall'}{\regext}}$.

    \item[Case] $\regexr=(\regexNeg\regexs)$, $\getFirst{\regexr}=\literalbigcap\{\inv{\literall}~|~\literall\in\getFirst{\regexs}\}~\cup~\{\literall\in\getFirst{\regexs}\}$:~ \\
      We obtain that $\lang{\nderiv{\literall}{\regexNeg\regexs}}=\regexWords\setminus\lang{\nderiv{\literall}{\regexs}}$.
      By induction $\forall\literall_\regexs\in\getFirst{\regexs},\literall_\regexs'\subset\literall_\regexs:$ $\lang{\nderiv{\literall_\regexs}{\regexs}}=\lang{\nderiv{\literall_\regexs'}{\regexs}}$.
      The chain holds because $\getFirst{\regexNeg\regexs}=\literalbigcap\{\inv{\literall}~|~\literall\in\getFirst{\regexs}\}~\cup~\{\literall\in\getFirst{\regexs}\}$ and $\lang{\nderiv{\literall'}{\regexNeg\regexs}}=\regexWords\setminus\lang{\nderiv{\literall'}{\regexs}}$ and for all $\literall''\not\in\getFirst{\regexs}$ $\nderiv{\literall''}{\regexs}=\emptyset$.

    \item[Case] $\regexr=(\regexs\regexConcatOp\regext)$:~
      \begin{description}
        \item[Subcase] $\isNullable{\regexs}$, $\getFirst{\regexr}=\getFirst{\regexs}\literaljoin\getFirst{\regext}$:~\\
          We obtain that $\lang{\nderiv{\literall}{\regexs\regexConcatOp\regext}}=\lang{\nderiv{\literall}{\regexs}\regexConcatOp\regext}\cup\lang{\nderiv{\literall}{\regext}}$.
          By induction $\forall\literall_\regexs\in\getFirst{\regexs},\literall_\regexs'\subset\literall_\regexs:$ $\lang{\nderiv{\literall_\regexs}{\regexs}}=\lang{\nderiv{\literall_\regexs'}{\regexs}}$ and $\forall\literall_\regext\in\getFirst{\regext},\literall_\regext'\subset\literall_\regext:$ $\lang{\nderiv{\literall_\regext}{\regext}}=\lang{\nderiv{\literall_\regext'}{\regext}}$.
          The chain holds because $\getFirst{\regexs\regexOrOp\regext}=\getFirst{\regexs} \literaljoin \getFirst{\regext}$ and $\forall\literall''\in\getFirst{\regexr}\cup\getFirst{\regexs}:~\exists\literall'''\in\getFirst{\regexs\regexOrOp\regext}:~\literall'''\subseteq\literall''$ and $\lang{\nderiv{\literall'}{\regexs\regexConcatOp\regext}}=\lang{\nderiv{\literall'}{\regexs}\regexConcatOp\regext}\cup\lang{\nderiv{\literall'}{\regext}}$.

        \item[Subcase] $\neg\isNullable{\regexs}$, $\getFirst{\regexr}=\getFirst{\regexs}$:~\\
          We obtain that $\lang{\nderiv{\literall}{\regexs\regexConcatOp\regext}}=\lang{\nderiv{\literall}{\regexs}\regexConcatOp\regext}$.
          By induction $\forall\literall_\regexs\in\getFirst{\regexs},\literall_\regexs'\subset\literall_\regexs:$ $\lang{\nderiv{\literall_\regexs}{\regexs}}=\lang{\nderiv{\literall_\regexs'}{\regexs}}$.
          The chain holds because $\getFirst{\regexs\regexConcatOp\regext}=\getFirst{\regexs}$ and $\lang{\nderiv{\literall'}{\regexs\regexConcatOp\regext}}=\lang{\nderiv{\literall'}{\regexs}\regexConcatOp\regext}$.

      \end{description}
  \end{description}
\end{proof}

\newpage
\section{Containment Example}
\label{sec:sample-calculation}

\begin{example}[Containment]\label{ex:containment}
  Consider the regular expressions $\regexr=(\regexOr{(\regexOr{\symbola}{\symbolb})}{\symbolc})$ and $\regexs=(\regexOr{\symbola}{\symbolb})$ and the inequality $\isSubSetOf{\regexr}{\regexs}$ which is obviously invalid. The computation of one derivation step is as follows:
  \begin{align}
    \isSubSetOf{\nderiv{\literall}{\regexr}}{\nderiv{\literall}{\regexs}}
    ~&\Leftrightarrow~ \isSubSetOf{\nderiv{\literall}{(\regexOr{(\regexOr{\symbola}{\symbolb})}{\symbolc})}}{\nderiv{\literall}{(\regexOr{\symbola}{\symbolb})}} \\
    ~&\Leftrightarrow~ \isSubSetOf
    {(\regexOr{\nderiv{\literall}{(\regexOr{\symbola}{\symbolb})}}{\nderiv{\literall}{\symbolc}})}
    {(\regexOr{\nderiv{\literall}{\symbola}}{\nderiv{\literall}{\symbolb}})}\\
    ~&\Leftrightarrow~ \isSubSetOf
    {(\regexOr{(\regexOr{\nderiv{\literall}{\symbola}}{\nderiv{\literall}{\symbolb}})}{\nderiv{\literall}{\symbolc}})}
    {(\regexOr{\nderiv{\literall}{\symbola}}{\nderiv{\literall}{\symbolb}})}
  \end{align}
  To solve the inequality $\isSubSetOf{\regexr}{\regexs}$ the inequality gets derived in respect to the next literals of $\isSubSetOf{\regexr}{\regexs}$. The calculation of $\getFirst{\isSubSetOf{\regexr}{\regexs}}$ is split into several sub-calculation concerning to the calculation of $\firstLiterals$.

  \begin{align}
    \getFirst{\regexr}
    ~&\Leftrightarrow~\getFirst{\regexOr{(\regexOr{\symbola}{\symbolb})}{\symbolc}}\\
    ~&\Leftrightarrow~\getFirst{\regexOr{\symbola}{\symbolb}} \literaljoin \getFirst{\symbolc}\\
    ~&\Leftrightarrow~(\getFirst{\symbola} \literaljoin \getFirst{\symbolb}) \literaljoin \getFirst{\symbolc}\\
    ~&\Leftrightarrow~(\{\symbola\} \literaljoin \{\symbolb\}) \literaljoin \{\symbolc\}\\
    ~&\Leftrightarrow~\{\symbola,\symbolb,\symbolc\}
  \end{align}

  \begin{align}
    \getFirst{\regexs}
    ~&\Leftrightarrow~\getFirst{\regexOr{\symbola}{\symbolb}}\\
    ~&\Leftrightarrow~\getFirst{\symbola} \literaljoin \getFirst{\symbolb}\\
    ~&\Leftrightarrow~\{\symbola\} \literaljoin \{\symbolb\}\\
    ~&\Leftrightarrow~\{\symbola,\symbolb\}
  \end{align}

  \begin{align}
    \getFirst{\regexNeg\regexs}
    ~&\Leftrightarrow~\literalbigcap\{\inv{\literall}~|~\literall\in\getFirst{\regexs}\}~\literaljoin~\{\literall\in\getFirst{\regexr}\}\\
    ~&\Leftrightarrow~\literalbigcap\{\inv{\symbola},\inv{\symbolb}\}~\literaljoin~\{\symbola,\symbolb\}\\
    ~&\Leftrightarrow~\{\inv{\{\symbola, \symbolb\}}\}~\literaljoin~\{\symbola,\symbolb\}\\
    ~&\Leftrightarrow~\{\inv{\{\symbola, \symbolb\}},\symbola,\symbolb\}
  \end{align}

  \begin{align}
    \getFirst{\isSubSetOf{\regexr}{\regexs}}
    ~&\Leftrightarrow~
    \getFirst{\regexAnd{\regexr}{\regexNeg\regexs}}\\
    ~&\Leftrightarrow~\{\literall\literalcap\literall'~|~\literall\in\getFirst{\regexr}, \literall'\in\getFirst{\regexNeg\regexs}\}\\
    ~&\Leftrightarrow~\{\literall\literalcap\literall'~|~\literall\in\{\symbola,\symbolb,\symbolc\}, \literall'\in\{\inv{\{\symbola, \symbolb\}},\symbola,\symbolb\}\}\\
    ~&\Leftrightarrow~\{\symbola,\symbolb,\symbolc\}
  \end{align}

  Finally, the inequality gets derived in respect to the next literals.

  \begin{align}
    \forall\literall\in\getFirst{\regexr}:~ \isSubSetOf{\nderiv{\literall}{\regexr}}{\nderiv{\literall}{\regexs}} ~|~ \getFirst{\regexr} = \{\symbola, \symbolb, \symbolc\}
  \end{align}

  This results in three iterations:

  \begin{align}
    \isSubSetOf{\nderiv{\symbola}{\regexr}}{\nderiv{\symbola}{\regexs}}
    ~&\Leftrightarrow~ \isSubSetOf{\nderiv{\symbola}{(\regexOr{(\regexOr{\symbola}{\symbolb})}{\symbolc})}}{\nderiv{\symbola}{(\regexOr{\symbola}{\symbolb})}} \\
    ~&\Leftrightarrow~ \isSubSetOf
    {(\regexOr{\nderiv{\symbola}{(\regexOr{\symbola}{\symbolb})}}{\nderiv{\symbola}{\symbolc}})}
    {(\regexOr{\nderiv{\symbola}{\symbola}}{\nderiv{\symbola}{\symbolb}})}\\
    ~&\Leftrightarrow~ \isSubSetOf
    {(\regexOr{(\regexOr{\nderiv{\symbola}{\symbola}}{\nderiv{\symbola}{\symbolb}})}{\nderiv{\symbola}{\symbolc}})}
    {(\regexOr{\nderiv{\symbola}{\symbola}}{\nderiv{\symbola}{\symbolb}})}\\
    ~&\Leftrightarrow~ \isSubSetOf
    {(\regexOr{(\regexOr{\regexEmpty}{\regexNull})}{\regexNull})}
    {(\regexOr{\regexEmpty}{\regexNull})}\\
    ~&\Leftrightarrow~ \isSubSetOf
    {(\regexOr{\regexEmpty}{\regexNull})}
    {\regexEmpty}\\
    ~&\Leftrightarrow~ \isSubSetOf
    {\regexEmpty}
    {\regexEmpty}
  \end{align}

  \begin{align}
    \isSubSetOf{\nderiv{\symbolb}{\regexr}}{\nderiv{\symbolb}{\regexs}}
    ~&\Leftrightarrow~ \isSubSetOf{\nderiv{\symbolb}{(\regexOr{(\regexOr{\symbola}{\symbolb})}{\symbolc})}}{\nderiv{\symbolb}{(\regexOr{\symbola}{\symbolb})}} \\
    ~&\Leftrightarrow~ \isSubSetOf
    {(\regexOr{\nderiv{\symbolb}{(\regexOr{\symbola}{\symbolb})}}{\nderiv{\symbolb}{\symbolc}})}
    {(\regexOr{\nderiv{\symbolb}{\symbola}}{\nderiv{\symbolb}{\symbolb}})}\\
    ~&\Leftrightarrow~ \isSubSetOf
    {(\regexOr{(\regexOr{\nderiv{\symbolb}{\symbola}}{\nderiv{\symbolb}{\symbolb}})}{\nderiv{\symbolb}{\symbolc}})}
    {(\regexOr{\nderiv{\symbolb}{\symbola}}{\nderiv{\symbolb}{\symbolb}})}\\
    ~&\Leftrightarrow~ \isSubSetOf
    {(\regexOr{(\regexOr{\regexNull}{\regexEmpty})}{\regexNull})}
    {(\regexOr{\regexNull}{\regexEmpty})}\\
    ~&\Leftrightarrow~ \isSubSetOf
    {(\regexOr{\regexEmpty}{\regexNull})}
    {\regexEmpty}\\
    ~&\Leftrightarrow~ \isSubSetOf
    {\regexEmpty}
    {\regexEmpty}
  \end{align}

  \begin{align}
    \isSubSetOf{\nderiv{\symbolc}{\regexr}}{\nderiv{\symbolc}{\regexs}}
    ~&\Leftrightarrow~ \isSubSetOf{\nderiv{\symbolc}{(\regexOr{(\regexOr{\symbola}{\symbolb})}{\symbolc})}}{\nderiv{\symbolc}{(\regexOr{\symbola}{\symbolb})}} \\
    ~&\Leftrightarrow~ \isSubSetOf
    {(\regexOr{\nderiv{\symbolc}{(\regexOr{\symbola}{\symbolb})}}{\nderiv{\symbolc}{\symbolc}})}
    {(\regexOr{\nderiv{\symbolc}{\symbola}}{\nderiv{\symbolc}{\symbolb}})}\\
    ~&\Leftrightarrow~ \isSubSetOf
    {(\regexOr{(\regexOr{\nderiv{\symbolc}{\symbola}}{\nderiv{\symbolc}{\symbolb}})}{\nderiv{\symbolc}{\symbolc}})}
    {(\regexOr{\nderiv{\symbolc}{\symbola}}{\nderiv{\symbolc}{\symbolb}})}\\
    ~&\Leftrightarrow~ \isSubSetOf
    {(\regexOr{(\regexOr{\nderiv{\symbolc}{\symbola}}{\nderiv{\symbolc}{\symbolb}})}{\nderiv{\symbolc}{\symbolc}})}
    {(\regexOr{\nderiv{\symbolc}{\symbola}}{\nderiv{\symbolc}{\symbolb}})}\\
    ~&\Leftrightarrow~ \isSubSetOf
    {(\regexOr{(\regexOr{\regexNull}{\regexNull})}{\regexEmpty})}
    {(\regexOr{\regexNull}{\regexNull})}\\
    ~&\Leftrightarrow~ \isSubSetOf
    {(\regexOr{\regexNull}{\regexEmpty})}
    {\regexNull}\\
    ~&\Leftrightarrow~ \isSubSetOf
    {\regexEmpty}
    {\regexNull}
  \end{align}

\end{example}

\end{document}